\AddToHook{begindocument/end}{%
  \makeatletter
  \def\@cite#1#2{[{#1\if@tempswa , #2\fi}]}%
  \makeatother
}
\PassOptionsToPackage{nocompress}{cite}
\documentclass[conference,compsoc]{IEEEtran}

\ifCLASSOPTIONcompsoc
    \usepackage{cite}
\else
  \usepackage{cite}
\fi

\ifCLASSINFOpdf
\else
\fi

\usepackage{xcolor} 
\usepackage{tcolorbox}  %
\usepackage{helvet}   %
\usepackage[T1]{fontenc}
\usepackage{amsmath}
\usepackage{amsfonts}
\usepackage{amsthm} %
\usepackage{algorithm}
\usepackage{algpseudocode}

\usepackage{hyperref} 
\usepackage{cleveref}

\usepackage{graphicx}
\usepackage{booktabs}
\usepackage{tikz}
\usepackage{nicefrac}

\usepackage{amssymb} %
\usepackage{pifont}  %

\newcommand{\TargetYes}{\tikz[baseline=-0.6ex]\fill (0,0) circle (3pt);} %
\newcommand{\TargetNo}{\tikz[baseline=-0.6ex]\draw (0,0) circle (3pt);}  %
\newcommand{\TargetBoth}{%
  \tikz[baseline=-0.6ex]{%
    \draw (0,0) circle (3pt);
    \fill[black] (0,0) ++(180:3pt) arc (180:0:3pt) -- cycle;
  }%
}

\newcommand{\Yes}{\checkmark}
\newcommand{\No}{\tikz[baseline=-0.6ex]{\node[scale=0.9]{\(\times\)};}} %

\newtheorem{definition}{Definition}

\newtheorem{theorem}{Theorem}
\newtheorem{lemma}{Lemma}
\newtheorem{assumption}{Assumption}
\newtheorem{corollary}{Corollary}

\newcommand{\vz}{\boldsymbol{z}}

\newcommand{\shortsection}[1]{\vspace*{1ex}\noindent{\bf #1.}}

\usepackage{booktabs}
\usepackage{textcomp}

\begin{document}

\makeatletter
\def\@cite#1#2{[{#1\if@tempswa , #2\fi}]}
\makeatother

\title{Beyond Indistinguishability: Measuring Extraction Risk in LLM APIs}

\author{
\IEEEauthorblockN{Ruixuan Liu}
\IEEEauthorblockA{Emory University\\
ruixuan.liu2@emory.edu}
\and
\IEEEauthorblockN{David Evans}
\IEEEauthorblockA{University of Virginia\\
evans@virginia.edu}
\and
\IEEEauthorblockN{Li Xiong}
\IEEEauthorblockA{Emory University\\
lxiong@emory.edu}
}

\maketitle

\begin{abstract}
Indistinguishability properties such as differential privacy bounds or low empirically measured membership inference are widely treated as proxies to show a model is sufficiently protected against broader memorization risks.
However, we show that indistinguishability properties are neither sufficient nor necessary for preventing data extraction in LLM APIs.
We formalize a privacy-game separation between extraction and indistinguishability-based privacy, showing that indistinguishability and inextractability are incomparable: upper-bounding distinguishability does not upper-bound extractability.
To address this gap, we introduce $(l,b)$-\emph{inextractability} as a definition that requires at least $2^{b}$ expected queries for any black-box adversary to induce the LLM API to emit a protected $l$-gram substring. 
We instantiate this via a worst-case extraction game and derive a rank-based extraction risk upper bound for targeted exact extraction, as well as extensions to cover untargeted and approximate extraction. 
The resulting estimator captures the extraction risk over multiple attack trials and prefix adaptations.
We show that it can provide a tight and efficient estimation for standard greedy extraction and an upper bound on the probabilistic extraction risk given any decoding configuration.
We empirically evaluate extractability across different models, clarifying its connection to distinguishability, demonstrating its advantage over existing extraction risk estimators, and providing actionable mitigation guidelines across model training, API access, and decoding configurations in LLM API deployment.
Our code is publicly available at: \url{https://github.com/Emory-AIMS/Inextractability}.

\end{abstract}

\IEEEpeerreviewmaketitle

\section{Introduction}
Large language models (LLMs) are increasingly consumed through hosted, black-box application programming interfaces (APIs)~\cite{openai_o1,deepseek_chat}, which are crucial for building agentic systems~\cite{wang2024survey} and domain applications~\cite{nori2023capabilities} at scale.
For practitioners, the API provides a ``prompt in, text out'' interface as shown in \Cref{fig:overview}.
While a few generation parameters can be selected by the user, the model's weights, training pipeline, and internal states remain opaque.
Even under this restricted access, some LLMs have been found to output text in their training data~\cite{nasr2023scalable,carlini2021extracting}, such as emails, code snippets, and long-form passages, which raises privacy, copyright infringement, and contextual integrity concerns.

Several methods have been proposed to measure data extraction risks in LLMs with varying extraction goals (targeted~\cite{carlini2022quantifying} or untargeted~\cite{carlini2021extracting}), adversary capabilities (discoverable~\cite{hayes2024measuring} or extractable~\cite{nasr2023scalable}), and evaluation criterion (exact match~\cite{carlini2022quantifying} or approximate match~\cite{ippolito2022preventing}).
The general concept of \emph{extractability} refers to the likelihood or the extent to which a model emits content from its training data.
The most representative  approach~\cite{carlini2022quantifying} identifies a segment as extractable when it can be greedily generated as a suffix by the target model when prompted with its prefix.

\begin{figure}
    \centering
    \includegraphics[width=0.9\linewidth]{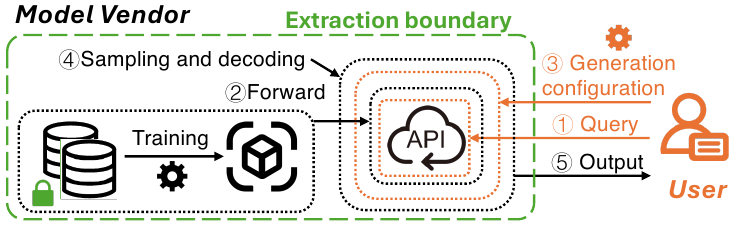}
    \caption{Overview of black-box LLM API.
    The adversary controls the prompt \textcircled{1} and some aspects of the generation \textcircled{3}. 
    Satisfying our $(l, b)$-inextractability definition requires that any black-box adversary needs at least $2^b$ expected API queries to emit a protected $l$-gram subsequence.
    }
    \label{fig:overview}
\end{figure}

Beyond heuristic defenses
including data deduplication~\cite{lee2021deduplicating}, post-training mitigation~\cite{chen2025parapo}, and output filtering~\cite{ippolito2022preventing}, differentially-private (DP) training~\cite{abadi2016deep} is often advocated as a defense because it provides a theoretical guarantee.
While satisfying the differential privacy definition provides an upper bound on the worst-case privacy risk in distinguishing neighboring training datasets, membership inference attacks (MIA)~\cite{shokri2017membership,hayes2025strong} can provide a lower bound on empirical privacy risk that is often used as a standard auditing tool~\cite{nasr2023tight,mahloujifar2024auditing}.
Both DP and MIA focus on \emph{distinguishability} which indicates the likelihood that the presence of one training record (which for language models is often modeled as a text fragment of some length) can be distinguished given the model.

A prevailing assumption underlying this use of DP and MIA is that there is a direct connection between \emph{distinguishability} and \emph{extractability}.
MIA has been used as a key component for data extraction~\cite{carlini2021extracting}, and MIA risk is believed to correlate with extraction risk~\cite{carlini2019secret}.
Several works~\cite{balle2022reconstructing,swanberg2025beyond} derive theoretical upper bounds on reconstruction success probability given a DP guarantee.
Given the reducibility of certain privacy games~\cite{salem2023sok}, mitigating the adversary advantage in the differential privacy distinguishability (DPD) game can also mitigate the advantage of other attacks, including data reconstruction (RC)~\cite{balle2022reconstructing}, membership inference (MI)~\cite{shokri2017membership}, and attribute inference (AI)~\cite{jayaraman2022attribute}.

However, this assumption may give a false sense of extraction risk.
Recent empirical results show that models can have weak MIA signals yet still regurgitate different forms of training data~\cite{liu2025language,ippolito2022preventing} or have high exposure risk~\cite{liu2025precurious}.
Thus, it remains an open question: what does \emph{indistinguishability} (as measured by a DP bound or low empirical MIA risk) imply for \emph{extractability} risk?

Our first goal is to clarify the relationship between \emph{indistinguishability} and \emph{inextractability}. 
Beyond previously noted contamination and ambiguous membership issues~\cite{duan2024membership}, we reveal that the extraction game is fundamentally separated from other inference privacy games.
While an inference adversary aims to distinguish whether a sample is in the training set (either from a data distribution prior or membership prior), an extraction adversary succeeds when the target model produces an explicit sequence which is not necessarily secret.
Therefore, the standard reducibility chain~\cite{salem2023sok} does not hold for data extraction: %
reducing \emph{distinguishability} risk does not necessarily ensure low \emph{extractability} risk.

We prove that indistinguishability and inextractability are neither sufficient nor necessary conditions for each other, and characterize the conditions under which the reducibility chain does not hold.
We further identify conditions under which the two risks appear to align, for example, when the signals used for MIAs are not calibrated by its sample hardness, clarifying what privacy measurements are informative for estimating extraction risk for an LLM API. %

To address the gap between indistinguishability and extraction risk, we introduce $(l,b)$-\emph{inextractability}, a definition which an LLM API satisfies if any black-box adversary requires at least $2^{b}$ expected queries to elicit any protected $l$-gram sequence. Our inextractability definition is inspired by and analogous to the differential privacy definition for indistinguishability.
While differential privacy bounds an adversary's capability to distinguish membership, ours bounds the capability of extraction.
We propose an \emph{extractability risk estimator} that accounts for multiple trials, prefixes, and decoding choices.
Instead of empirical generation or relying on specific sampling, we estimate the risk via the rank of the ground truth token's sampling probability.%

\shortsection{Contributions}
Compared to previous empirical extraction measurement methods~\cite{carlini2022quantifying,nasr2023scalable,hayes2024measuring,ippolito2022preventing,biderman2023emergent} that depend on specific generation configurations, our cost-based measurement provides an empirical upper bound on extraction risk given any generation configuration. It provides an efficient estimate for the standard greedy-decoding extraction~\cite{carlini2022quantifying} and upper bounds the recent probabilistic extraction measurement~\cite{hayes2024measuring}.
Beyond exact and targeted extraction, we extend the bound to untargeted and approximate extractions, enabling a unified extraction risk estimation. 
Our key contributions are:
\begin{enumerate}
	\item Formalizing a privacy-game separation between extraction and classic inference privacy games, showing that indistinguishability and inextractability are not reducible, and we characterize  the conditions under which they align (\Cref{sec:pre_unbounded}).
	\item Analyzing the limitations of existing extraction measurements (\Cref{sec:pre_underestimate}), and introducing a definition of $(l,b)$-inextractability  
    for LLM APIs (\Cref{sec:method_definition}) along with a worst-case upper-bound estimator based on rank probabilities (\Cref{sec:method_worst}).
    \item Demonstrating the generality of our measurement by providing an efficient estimator for greedy generation and by upper-bounding the recent probabilistic measurement (\Cref{sec:method_connection}). We further extend the proposed guarantee to untargeted and approximate extraction (\Cref{sec:method_beyond}).
	\item Empirically studying the relationship  between indistinguishability and inextractability, showing the extraction risk upper bound and the properties of our measurement with both pre-trained and fine-tuned LLMs (\Cref{sec:experiment}).
	\item Outlining and evaluating the effectiveness of mitigation strategies including training, API access control and decoding configuration (\Cref{sec:defense}).
\end{enumerate}

\section{Unbounded and Underestimated Risks}\label{sec:challenge}
This section provides background on LLM pipelines (Section \ref{sec:preliminary}) and our LLM APIs threat model (Section \ref{sec:pre_threatmodel}).
Then we show that 
extraction risks are misaligned with existing distinguishability risks (Section \ref{sec:pre_unbounded}) and standard extraction measurements may underestimate the risk due to their limited generation settings (Section \ref{sec:pre_underestimate}).

\subsection{Preliminaries}\label{sec:preliminary}
An LLM, $f_\theta$, is parameterized by weights $\theta$ and trained on a dataset $D = \{z^{(1)}, z^{(2)}, \dots, z^{(n)}\}$ where each $z^{(i)} = (z_1^{(i)}, z_2^{(i)}, \ldots, z_{L_i}^{(i)})$ is a tokenized sequence from vocabulary $\mathcal{V}$. 
At execution time, $f_\theta(v \mid z_{<t}^{(i)})$ denotes the logits assigned by $f_\theta$ to token $v \in \mathcal{V}$ given the preceding context $z_{<t}^{(i)}$.
\Cref{tab:notation} summarizes notations we use throughout the paper.

\begin{table}[tb]
\centering
\caption{Summary of notation.}\label{tab:notation}
{
\begin{tabular}{cp{2.5in}}
\toprule
\textbf{Notation} & \textbf{Description} \\
\midrule
$\mathcal{T}$ & Training algorithm \\
$f_\theta$ & Model parameterized by $\theta$ \\
$\phi_{k, T}$ & Decoding mechanism with parameters $k$ and $T$ \\ 
$P_t(v)$ & Sampling probability for token $v\in\mathcal{V}$ at position $t$ \\
$\circ$ & Concatenation of text or composition of functions \\$\mathcal{O}_{\theta\circ \phi}$ & Model API oracle, including model $f_\theta$ and decoder $\phi$ \\
$m$ & Number of accessible top token probabilities \\
$\mathcal{A}$ & Extraction mechanism \\
$\mathcal{S}$ & Valid extraction candidate set \\
$D$ & Training dataset $\{z^{(1)}, z^{(2)}, \ldots, z^{(N)}\}$ \\
$D_\text{pro}$ & Protected dataset, assumed to be a subset of $D$ \\
$l$ & Length of sub-sequence \\
$b$ & Inextractability level; $2^b$ expected queries to oracle $\mathcal{O}$ \\
$n$ & Number of extraction attempts \\
$p$ & At-least-once extraction success rate over $n$ trials \\
$p_{\mathbf{z}}$ & Single-trial success probability for text $\mathbf{z}$\\
$p_{\mathbf{z}}^*$ & Maximum single-trial success probability for text $\mathbf{z}$ \\
$r_t$ & Rank over $\mathcal{V}$ of ground-truth token at position $t$ \\
$\mathcal{L}$ & Matching criterion for a successful extraction \\
$(\epsilon, \delta)$ & Privacy loss parameters for  differential privacy \\
\bottomrule
\end{tabular}
}
\end{table}

\shortsection{Training}
LLMs are trained in an auto-regressive manner to maximize the likelihood of each token \( z_t \) given its preceding context \( z_{<t} \) over all training sequences.

\shortsection{Inference} 
Given a query $\vz_p$, the model returns the response $\vz_s$ by forwarding the query to the model $f_\theta$ and processing its output via a decoding mechanism $\phi$.
We summarize a general pipeline in \Cref{alg:decode} based on truncation sampling~\cite{hewitt2022truncation}. The pipeline takes as input the logits $\ell = f_\theta(\cdot \mid z_{<t})$ and the decoding parameters, and outputs the next token. This pipeline covers the decoding mechanisms commonly used in LLM APIs.
For example, greedy decoding is described by $k=1$; top-$k$ selects $\mathcal{V}_\text{top}$ as tokens with top-$k$ $P_t(v)$; top-$p$ truncates by $\sum_{v\in \mathcal{V}_{top}}P_t(v)$.

\renewcommand{\algorithmicrequire}{\textbf{Inputs:}}
\renewcommand{\algorithmicensure}{\textbf{Output:}}
\begin{algorithm}[H]
\caption{General Pipeline of Truncation Sampling}
\label{alg:decode}
\begin{algorithmic}[1]
\Require Logits $\ell$, temperature $T$, selected $\mathcal{V}_\text{top}$ with size $k$.
\State $P_t(v) \gets \mathrm{Softmax}(\ell/T)_v$ \textcolor{gray}{// Temperature-adjusted}
\State $\tilde{P}_t(v) \gets P_t(v)$ if $v \in \mathcal{V}_{\text{top}}$, else $0$ \textcolor{gray}{// Truncated}
\State $\hat{P}_t(v) \gets \tilde{P}_t(v) / \sum_{u \in \mathcal{V}} \tilde{P}_t(u)$ \textcolor{gray}{// Re-normalized} \label{line:re-norm}
\Ensure A token sampled from $\hat{P}_t(v)$
\end{algorithmic}
\end{algorithm}

\begin{table}[tb]
\centering
\caption{%
Output accessibility and decoding control.
}\label{tab:llm_access_control}
{
\begin{tabular}{lcc}
\toprule
\multicolumn{1}{c}{\textbf{Model API}} & \textbf{Probabilities}
& \textbf{Decoding Control} \\
\midrule
OpenAI GPT-5~\cite{openai_api_gpt5}             & No                      & No \\
Anthropic Claude~\cite{anthropic_openai_sdk}      & No                      & Yes \\
Azure OpenAI API~\cite{azure_openai} & Top-5 & Yes \\
OpenAI GPT-4~\cite{openai_chat_api}      & Top-20                  & Yes \\
Gemini~\cite{google_vertex_ai_generative}            & Top-20  & Yes \\
\midrule
Evaluated API  & Top-$m$ & Yes \\
\bottomrule
\end{tabular}
}
\end{table}

\subsection{Threat Model}\label{sec:pre_threatmodel}

We focus on a practical threat model for production LLMs where only an inference API is available, as illustrated in \Cref{fig:overview}.
The model vendor trains the model with training method $\mathcal{T}$ using internal dataset $D$ and releases the inference API oracle, $\mathcal{O}_{f_\theta\circ\phi_{k, T}}$.

\shortsection{Adversary Capability}
\Cref{tab:llm_access_control} summarizes two key aspects of access available through existing LLM APIs.
Instead of previous full-probability~\cite{shi2023detecting} or text-only access~\cite{nasr2023scalable},
we consider a general and realistic setting where the model vendor provides output probabilities for the top-$m$ tokens.
Typically the API in common deployments~\cite{openai_chat_api} limits the number output probabilities return to a maximum of $m=20$. 
We can cover the range of assumptions from $m=0$ (only receiving the predicted text) to $m=|\mathcal{V}|$ (full access to all logits).
In addition, some production APIs provide users with control over decoding hyperparameters such as the temperature $T$ and the threshold of top-$k$/top-$p$ for determining the selected token set $\mathcal{V}_\text{top}$ in \Cref{alg:decode}.
For risk measurement, we consider the general setting where the top-$m$ tokens with sampling probabilities are available and auditors can control the decoding parameters in $\phi$.

\shortsection{Adversary Goals}
The adversary's goal is to extract text fragments in the protected subset of the training data $D_\text{pro}$. 
Each element $\vz\in D_\text{pro}$ is a sequence of tokens; within each sequence, any contiguous $l$-gram $\vz_{t:t+l}$ is considered a protected subsequence.
This concept covers copyrighted paragraph with larger $l$ and pre-defined PIIs with shorter $l$.
We provide a definition of data extraction in \Cref{def:extraction}, which covers various data extraction attacks considered in previous works~\cite{carlini2021extracting, carlini2022quantifying, nasr2023scalable}:

\begin{definition}[\textbf{Data Extraction}]\label{def:extraction}
Given the model oracle $\mathcal{O}$ and a defender-specified extraction set $\mathcal{S}$ representing the protected content under a given matching criterion, 
\textbf{data extraction} occurs when (1) 
a text is generated by the oracle with some input as $\hat{\vz}_s=\mathcal{O}(\vz_p)$, and (2) the generated text is in the extraction set, $\hat{\vz}_s\in \mathcal{S}$.
\end{definition}

For \emph{exact} data extraction~\cite{carlini2021extracting}, the extraction set includes exact protected training data, $\mathcal{S} \subseteq D_\mathrm{pro}$; for \emph{approximate} data extraction~\cite{ippolito2022preventing}, the extraction set, $\mathcal{S} \subseteq \tilde{D}_\mathrm{pro}$, includes texts that approximately match the training data.
Specifically, 
given the matching criterion $\mathcal{L}$, the exact match is $\mathcal{L}(\hat{z}, z)=\mathbb{I}\{\hat{\vz}=z\}$, and the approximate match is $\mathcal{L}(\hat{\vz}, z)=\mathbb{I}\{d(\hat{\vz}, z)\leq c\}$ where $d(\cdot)$ is a distance metric (e.g., $1-$BLEU or $1-$ROUGE-L).

For \emph{targeted} data extraction~\cite{carlini2022quantifying}, the extraction set $\mathcal{S}=\{\vz_s\}$ for a target $\vz=\vz_p\circ\vz_s$ only contains the original suffix associated with its prefix and $\vz\in D_\mathrm{pro}$.
For \emph{untargeted} data extraction~\cite{carlini2021extracting}, $\mathcal{S}=\{\vz\,|\,\vz\in D_\mathrm{pro}, |\vz|=l\}$ includes all possible \emph{l}-gram sequences in the protected training dataset $D_\mathrm{pro}$.
We focus on exact and targeted extraction and assume $D_\text{pro}=D$ for simplicity.
We extend to approximate and untargeted data extraction in later sections.

\begin{table*}[t]
\centering
\caption{Privacy game comparison. \\ {\rm  The top four games are \emph{distinguishability} games, in contrast to data extraction. $\mathcal{A}^\prime$ denotes the secret choosing process in the worst-case attack.}}\label{tab:game}
\resizebox{\linewidth}{!}
{
\begin{tabular}{cccccc}
\toprule
{\bf Privacy Game} & {\bf Secret Construction} & {\bf Training} & {\bf Attack} & {\bf Posterior} ($p_1$) & {\bf Prior} ($p_0$) 
\\
\midrule
DP Distinguishability (DPD) & $D, b\in\{0, 1\}, \vz_0, \vz_1 \leftarrow \mathcal{A}^\prime (\mathcal{T}, n)$ & $\theta\leftarrow \mathcal{T}(D\cup \{\vz_b\})$ & $\hat{b}\leftarrow \mathcal{A}(\mathcal{T}, \theta, D, \vz_0, \vz_1)$ & $\Pr[\mathcal{A}:\hat{b}=b]$ & 1/2 
\\
Reconstruction (RC) & $D \sim \mathcal{D}^{N-1}, \vz\sim \pi$ & $\theta\leftarrow \mathcal{T}(D\cup \{\vz\})$ & $\hat{\vz} \leftarrow \mathcal{A}(\mathcal{T}, \theta, D)$ & $\Pr[\mathcal{A}:\hat{\vz}=\vz]$ & $\Pr[z, \hat{z} \sim \pi:\hat{\vz}=\vz]$ \\
Membership Inference (MI) & $D\sim \mathcal{D}^{N-1}, b\in\{0, 1\}, \vz_0, \vz_1 \sim \mathcal{D}$ & $\theta\leftarrow \mathcal{T}(D\cup \{\vz_b\})$ & $\hat{b}\leftarrow\mathcal{A}(\mathcal{T}, \mathcal{D}, \theta, n, \vz_0)$ & $\Pr[\mathcal{A}:\hat{b}=b]$ & 1/2 %
\\ %
Attribute Inference (AI) & $D\sim \mathcal{D}^N, b\in \{0, 1\}, \vz_0 \sim D, \vz_1 \sim \mathcal{D}$ & $\theta\leftarrow \mathcal{T}(D)$ & $\hat{z}_s \leftarrow \mathcal{A}(\mathcal{T}, \mathcal{D}, n, \theta, \vz_p)$ & $\Pr[\mathcal{A}: \hat{\vz}_s=\vz_s]$ & $\Pr[\mathcal{D}: \hat{\vz}_s=\vz_s]$ %
\\
\midrule
{\bf Data Extraction (DE)} & $D\sim \mathcal{D}^N, \vz=\vz_p\circ \vz_s \sim D$ & $\theta\leftarrow \mathcal{T}(D)$ & $\hat{\vz}_s\leftarrow \mathcal{A}(\mathcal{O}_{\theta\circ\phi}(\cdot), \vz_p)$ & $\Pr[\mathcal{A}:\hat{\vz}_s=\vz_s]$ & N/A %
\\ %
\bottomrule
\end{tabular}
}
\end{table*}

\subsection{Extraction is not Bounded by Inference Risk}\label{sec:pre_unbounded}

Other privacy notions and attacks are often considered to be useful proxies for extraction risk, but we will show that they have independent privacy games that are not comparable with extraction risk.

\shortsection{Reducibility of distinguishability privacy games}
Common privacy notions include differential privacy distinguishability (DPD), membership inference (MI), reconstruction (RC) and attribute inference (AI).
As summarized in \Cref{tab:game}, all of these focus on the \textit{distinguishability} of a given secret: 
(1) DPD and MI distinguish between two models trained with $\vz_0$ or $\vz_1$;
(2) AI distinguishes if the attribute $\hat{\vz}_s$ is predicted given the target model or drawn from the distribution;
(3) RC distinguishes if an output $\hat{\vz}$ is generated from the target model or a prior secret distribution $\pi$.

Salem et al.\ established a privacy reducibility chain~\cite{salem2023sok} among these four notions:
$$
{\text{DPD}\preceq \text{RC} \preceq \text{MI} \approx \text{AI}}
$$
where $G_1 \preceq G_2$ means the privacy game $G_1$ is reducible to $G_2$ and is at most as hard to win as $G_2$ with $\text{Adv}_{G_1} \geq c \cdot \text{Adv}_{G_2}$ given some constant $c \gtrsim 1$.
The attack advantage $\text{Adv}=(p_1-p_0)/(1-p_0)$ is the normalized gain on the success rate $p_1$ compared to the prior $p_0$.
Thus, a defense that can defend against $G_1$ by establishing $\text{Adv}_{\text{G}_1}<a$ is also effective against $G_2$ as $\text{Adv}_{\text{G}_2}<a/c$.
Hence, DP training is assumed to defend against other attacks, because DP ensures worst-case bounds for the DPD privacy game.
Recent theoretical analyses~\cite{swanberg2025beyond,balle2022reconstructing} establish the DP bound on reconstruction (RC) risk as in \Cref{theo:dp}:

\begin{theorem}[Reconstruction Bound from DP~\cite{swanberg2025beyond}]\label{theo:dp}
Let $\mathcal{T}(D)$ be a training algorithm that satisfies $(\epsilon,\delta)$-DP, with training dataset $D \sim \mathcal{D}^N$, output $\theta \subseteq \mathrm{Supp}(\mathcal{T})$.
Denoting a specific attack guess is $\hat{z}=\mathcal{A}(\theta)$, then for any output $\theta$ and any training sample $z\in D$ we have:
$$
\Pr_{\substack{z\in D}}[\mathcal{L}(\hat{z}, z)=1] \leq \frac{e^\epsilon}{e^\epsilon-1+\frac{1}{\Pr_{\substack{z\sim\mathcal{D}}}[\mathcal{L}(\hat{z}, z)=1]}} + \delta.
$$
\end{theorem}

\noindent 
Essentially, DP bounds $p_1/p_0$, the ratio between the reconstruction posterior given the target model $f_\theta$ or a known prior distribution $\pi$. (A slightly looser bound of $e^\epsilon \cdot \Pr_{z\sim\mathcal{D}}[\mathcal{L}(\hat{z}, z)=1]+\delta$ is proven by Cohen et al.~\cite{cohen2024data}.)

\begin{figure}[bt]
  \includegraphics[trim=0 0 0 0, clip, width=0.9\linewidth]{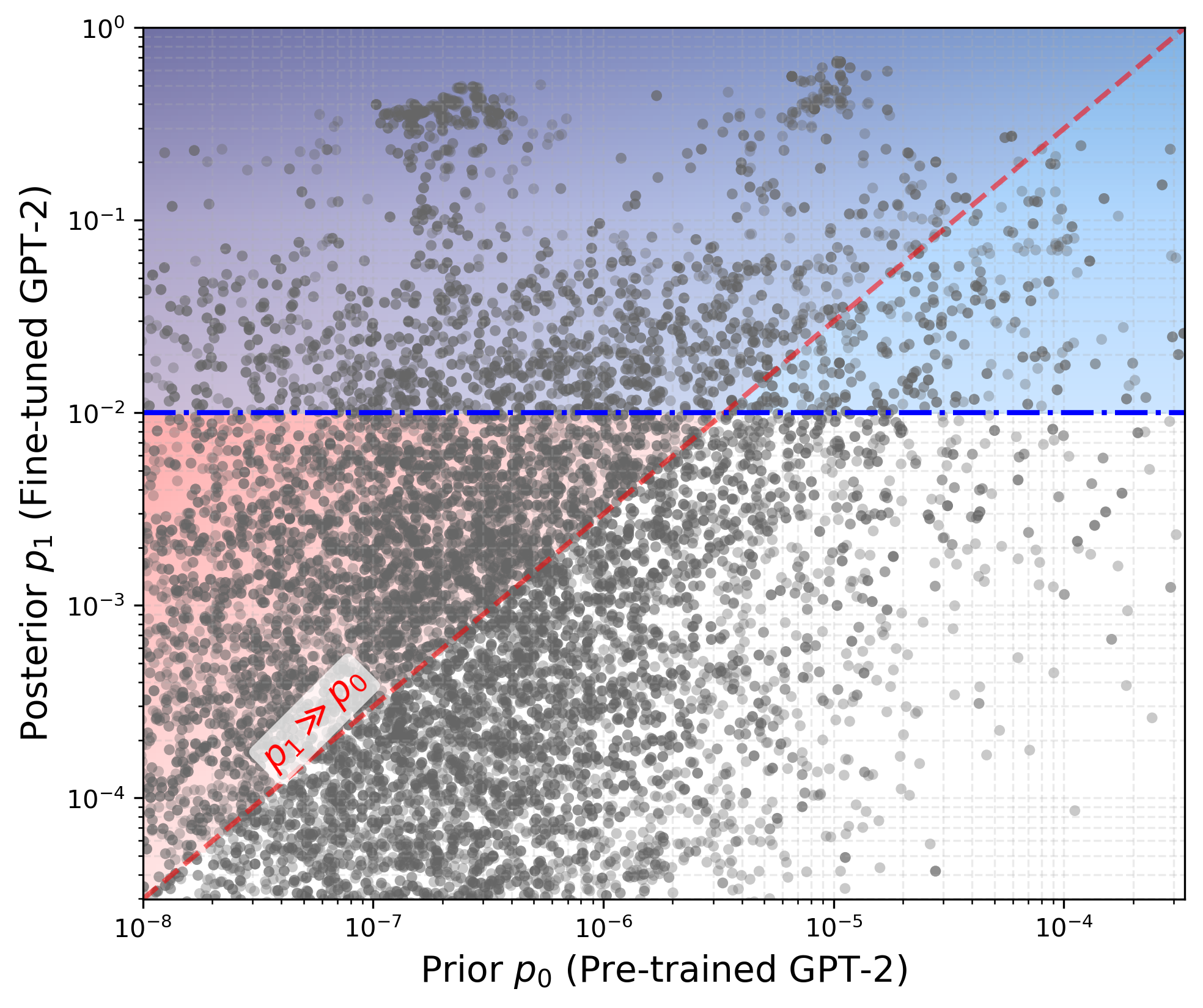}
  \caption{Partially independent risks of distinguishability (high posterior advantage over prior $p_1\gg p_0$, reddish area) and extractability (high posterior $p_1$, blue area). Each dot represents one short email fragment from the Enron dataset using the generation probability on fine-tuned GPT-2 as posterior and following Swanberg et al.~\cite{swanberg2025beyond} to compute the prior on GPT-2.
  Points in the blue triangle show that distinguishability is \textbf{unnecessary} for extractability---the extractability risk is high, but the reconstruction risk is not significant given the relatively small advantage $p_1/p_0$;
  points in the reddish area show that distinguishability is \textbf{insufficient} for extractability---the reconstruction advantage is significant, but the extraction risk is low.}\label{fig:dist_ext}
\end{figure}

\shortsection{Separated privacy game of extractability}
However, data extraction (DE) in \Cref{def:extraction} has non-trivial distinctions between above privacy games, even for reconstruction (RC) which has the similar goal to extract or recover training samples.
To highlight the distinctions between data extraction and reconstruction, we plot the probability of extracting a short fragment from a prior (GPT-2) or a model (posterior) in \Cref{fig:dist_ext}.
We note that reconstruction focuses the relative increase of the posterior probability ($p_1$) over the prior probability ($p_0$) (reddish area), similar to other distinguishability games and the ratio that is at the core of the differential privacy definition.
In contrast, data extraction focuses on the absolute posterior probability ($p_1$, shaded blue in \Cref{fig:dist_ext}).

Therefore, we need to separate the \textit{extractability} game with \textit{distinguishability} games.
The separation matters for two reasons:
1) The practical threat of verbatim generation is not limited to leaking private information, but also includes the risk that the model would output copyrighted material or violate contextual integrity. A protected content with a high prior (e.g., a famous copyrighted book) does not indicate that it can be freely emitted.
2) Bounding the advantage of the posterior, as in \Cref{theo:dp}, does not indicate concrete extraction risk level, especially when the prior cannot be accurately estimated as is typical for natural language.

\begin{figure*}
\centering
    \includegraphics[width=\linewidth]{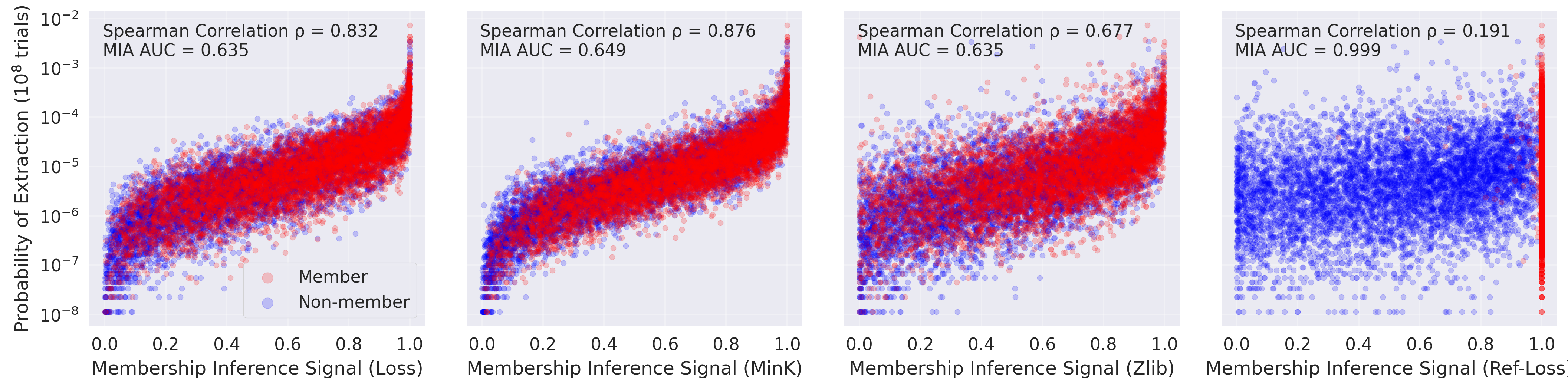}
	
\caption{
Observed correlation between distinguishability (as measured using four membership inference attacks) and extractability.}
\label{fig:mia_ext_sample}
\end{figure*}

\shortsection{Implications of separated privacy games}
Clarifying the differences between these distinct privacy games helps resolve common misunderstandings about extractability defenses and attacks and explains previous counter-intuitive observations in the literature.

\noindent\textit{1) A false sense of extractability risk.}
While DP training bounds distinguishability risk, it does not constrain extractability risk.
Even if a model exhibits low risk under distinguishability-based privacy games (including DPD, RC, MI and AI), its extractability risk may remain high.
For example, a DP-trained model that has negligible membership inference against a strong adversary can still output high posterior training data even with a strict privacy budget~\cite{liu2025precurious}. We formalize this in \Cref{theo:dpd_de}, showing that a low distinguishability risk does not necessarily indicate a low extractability risk when $p_0$ is large or $\varepsilon$ is small.

\begin{theorem}[$\text{DPD} \not\preceq \text{DE}$]\label{theo:dpd_de}
Resilience of DPD does not imply the resilience of data extraction.
\end{theorem}
\textit{Proof.} Let a training pipeline $\mathcal{T}$ satisfy $(\epsilon, \delta)$-DP, and thus be resilient against DPD as defined in \Cref{tab:game}.
The DPD advantage is bounded by the privacy budget as $\text{Adv}_\text{DPD}(\mathcal{A})\leq \frac{e^\epsilon-1+2\delta}{e^\epsilon+1}$~\cite{humphries2023investigating}.
Since extractability focuses on the posterior, $\text{Adv}_\text{DE}(\mathcal{A})=p_1 \leq \frac{e^\epsilon}{e^\epsilon-1+1/p_0}+\delta$ given  \Cref{theo:dp} with the prior $p_0=\Pr_{\substack{z\sim\mathcal{D}}}[\mathcal{L}(\hat{z}, z)=1]$.
Setting $\delta=0$, DPD cannot be reduced to DE when
\begin{align}
\frac{e^\epsilon - 1}{e^\epsilon + 1} &< 
\frac{e^\epsilon}{e^\epsilon - 1 + \frac{1}{p_0}}
\end{align}

\noindent
This inequality holds if and only if $p_0 > \frac{e^\epsilon - 1}{3e^\epsilon - 1}$. 
For \textbf{fixed $\epsilon$}, choose $p_0 > \frac{e^\epsilon - 1}{3e^\epsilon - 1}$; since this threshold approaches $\frac{1}{3}$ as $\epsilon \to \infty$, setting $p_0 > \frac{1}{3}$ suffices for all $\epsilon > 0$.
For \textbf{fixed $p_0$}, the inequality holds for all $\epsilon > 0$ when $p_0 > \frac{1}{3}$, while $p_0 < \frac{1}{3}$ requires $\epsilon < \ln\left(\frac{1-p_0}{1-3p_0}\right)< 1.94$.

\vspace{1ex}
\noindent\textit{2) A false sense of distinguishability risk.}
Similarly, we show in \Cref{theo:de_dpd} that a model that shows low extractability risk does not necessarily have low distinguishability risk.
A recent experimental work~\cite{hayes2025strong} supports this claim by showing that samples with low extraction risk are still distinguishable by membership inference attacks.

\begin{theorem}[$\text{DE} \not\preceq \text{DPD}$]\label{theo:de_dpd}
Resilience to DE does not imply a bound on distinguishability risk of DPD.
\end{theorem}
\textit{Proof.} 
Similarly, 
given above upper bounds for $\text{Adv}_\text{DE}$ and $\text{Adv}_\text{DPD}$, there are two cases where the reducibility condition of $\text{DE} \preceq \text{DPD}$ does not hold: 1) the prior $p_0\ll 1/3$; and 2) fixing $p_0$, the privacy budget $\epsilon>\ln\left(\frac{1-p_0}{1-3p_0}\right)$.

\shortsection{Explaining apparently-inconsistent observations}\label{ssec:explaininginconsistencies}
It is commonly assumed that there is a direct connection~\cite{carlini2021extracting,carlini2019secret} between MIA and extraction, such as using sample-level MIA signal to identify data extraction. However, prior works have shown variation between membership inference and extraction risks.
For example, recent works~\cite{hayes2025strong,liu2025precurious} observe that distinguishability risk does not indicate extraction risk.

We find that both observations essentially align with our \Cref{theo:dpd_de} and \Cref{theo:de_dpd}, although they appear conflicting under certain conditions.
The main cause behind such confusion is the membership inference attack signals.
We show the interplay between membership inference (MI) risk and data extraction (DE) risk in \Cref{fig:mia_ext_sample} on a fine-tuned GPT-2 for four different MIA signals:
Loss~\cite{yeom2018privacy}, MinK~\cite{shi2023detecting}, Zlib~\cite{carlini2021extracting} and Ref-Loss~\cite{carlini2022membership}.
Different MIA signal are normalized to the same scale, and we compute the extraction probability using the method from Hayes et al.~\cite{hayes2024measuring}.
There is a clear trend that the correlation between MI and DE is stronger for non-calibrated signals (Loss, MinK) and weaker for calibrated signals (Zlib, Ref-Loss).

This sheds light on seemingly conflicting observations in prior work. The direct connection between MIA and DE is discussed using a loss-based MIA  signal~\cite{carlini2019secret}. This occurs because non-calibrated signals (Loss, MinK) reflect the model's absolute posterior confidence, which directly aligns with the extraction objective. In contrast, lower correlation~\cite{hayes2025strong} 
is observed with the calibrated  signals (Zlib, Ref-Loss). The calibration normalizes the confidence via a sample hardness factor such as the zlib entropy or reference-model perplexity, leading to the relative confidence improvement (posterior over prior) in generating a sample, which aligns with the MIA objective. Hence, MIA can appear correlated or uncorrelated with extraction risk, even though both originate from the same underlying model behavior.

\begin{table*}[t]
\centering
\setlength{\tabcolsep}{5pt}
\caption{Comparison of extraction or memorization evaluation methods. \\ {\rm 
Goal symbols: \protect\TargetYes\; targeted, \protect\TargetNo\; untargeted, \protect\TargetBoth\; both. 
Tier indicates the relative bound category (Upper or Lower) reported by the original works.
}}
\label{tab:extraction}
\renewcommand{\arraystretch}{1.12}
\resizebox{\textwidth}{!}{%
\begin{tabular}{cccccccccc}
\toprule
Method & Goal & Approx. & Multiple Trials & \multicolumn{1}{c}{Prefix} & 
\multicolumn{1}{c}{Decoding} & 
\multicolumn{1}{c}{Metric} & 
\multicolumn{1}{c}{Estimated Risk} \\
\midrule
\cite{carlini2021extracting}       & \TargetNo            & \No          & \Yes & No prefix / Internet snippets            & Beam-search + Temperature          & $k$-Eidetic                 & Lower bound \\
\cite{nasr2023scalable}            & \TargetNo            & \No           & \Yes & Internet snippets / Repeated tokens     & Limited by APIs         & $n$-gram match              & Lower bound \\
\cite{ippolito2022preventing}      & \TargetNo            & \No/\Yes        & \No  & Working                  & Greedy               & BLEU                 & Lower bound \\
\cite{hayes2024measuring}          & \TargetYes           & \No/\Yes        & \Yes & Original                 & Top-k + Temperature         & $(n,p)$-discoverable        & Upper bound \\
\cite{carlini2022quantifying}      & \TargetBoth          & \No          & \No  & Original                 & Greedy / Beam-search        & $k$-extractable             & Upper bound \\
\cite{biderman2023emergent}        & \TargetYes           & \No          & \No  & Original                 & Greedy               & Memorization score        & Upper bound \\
\cite{pappu2024measuring}          & \TargetYes           & \Yes           & \No  & Original                 & Greedy               & Counterfactual memorization & Upper bound \\
\cite{kassem2024alpaca}            & \TargetYes           & \Yes         & \No  & Optimized (Init. w/original prefix)        & Greedy               & ROUGE-L              & Upper bound \\
\cite{schwarzschild2024rethinking} & \TargetYes           & \No           & \No  & Optimized (Init. w/random prefix)       & Greedy               & Compression ratio             & Lower bound \\
\hline
Ours & \TargetBoth & \No / \Yes & \Yes & Empirically optimal (w/all original) & Optimal decoding & ($l, b$)-inextractability  & Upper bound \\
\bottomrule
\end{tabular}%
}
\end{table*}

\subsection{Suboptimal Attacks Underestimate Risk}\label{sec:pre_underestimate}
From the above discussion, only data extraction attacks directly audit extractability.
As summarized in~\Cref{tab:extraction}, previous extraction attacks either estimate a \textit{lower bound} (i.e., extractable)~\cite{nasr2023scalable} or \textit{upper bound} (i.e.,  discoverable)~\cite{hayes2024measuring} on extraction risk.
We focus on the worst-case risk (i.e., upper bound) because a single extraction can cause irreversible harm regardless of average-case behavior. Moreover, bounding the worst-case risk provides a rigorous and distribution-independent metric for evaluating and comparing defense mechanisms.
However, current upper-bound estimation methods are prone to underestimating the risk due to suboptimal attack strategies.

\shortsection{Extraction is a discoverable game}
Most data extraction attacks only conduct one time generation and estimate the extraction risk given the binary result of extraction or not.
However, the extractability is a \emph{discoverable game}~\cite{hayes2024measuring} where the chance of a successful extraction increases when attackers can query more times.
For example, given the extraction probability of one sample as $p_{\mathbf{z}}$, the probability that it can be extracted after $n$ independent trails is $1-(1-p_{\mathbf{z}})^n$.

\shortsection{Suboptimal prefix for discoverable extraction}
Text drawn from the internet or divergent prompts (e.g., ``poem poem poem'') can serve as prefixes to extract training data~\cite{nasr2023scalable}, but these methods primarily estimate the extraction risk (i.e., the lower bound) assuming the adversary has no access to the target text.
To assess the discoverable extraction risk (i.e., the upper bound) from a conservative defender's perspective, prior works~\cite{hayes2024measuring,carlini2022quantifying,biderman2023emergent,pappu2024measuring} assume that the complete training sample is known and the original prefix is the most effective trigger to generate its suffix. 
However, these works use a fixed length prefix, and do not consider extending the prefix adaptively. %
Other works~\cite{kassem2024alpaca,schwarzschild2024rethinking} use adversarially-optimized prefixes to improve the likelihood of generating a suffix, but the optimality is not guaranteed and the optimization requires extra trials.

\shortsection{Suboptimal decoding for discoverable extraction}
As summarized in~\Cref{tab:extraction}, most prior works use greedy decoding to generate outputs, effectively setting $|\mathcal{V}_\text{top}|=1$ in \Cref{alg:decode}. A few works~\cite{nasr2023scalable,hayes2024measuring} use specific decoding configurations with fixed temperature or top-$k$ settings.
Beam-search heuristically searches for the highest likelihood generated sequence~\cite{carlini2021extracting,carlini2022quantifying}, but the beam width is capped and performing beam-search requires extra queries.
Evidence in previous results~\cite{carlini2022quantifying,hayes2024measuring,yu2023bag} indicates that using non-greedy decoding sometimes slightly increases the extraction rate.
However, any fixed decoding hyperparameters are unlikely to be optimal, leaving an unknown gap between the empirical estimate and actual worst-case extraction risk. 

\section{Black-Box Extraction Risk}\label{sec:method}

To avoid misleading conclusions from distinguishability-based guarantees, 
we propose a direct definition of inextractability (\Cref{sec:method_definition}).
To avoid underestimating risk, our definition provides a formal upper bound on worst-case data extraction risk (\Cref{sec:method_worst}) via a Data Extraction game analogous to the DPD game for the differential privacy definition. Then we discuss connections to existing generation-based and probabilistic measurements (\Cref{sec:method_connection}). In \Cref{sec:method_beyond}, we extend to other extraction variants such as untargeted extraction and approximate matching criterion.

As summarized in \Cref{tab:extraction}, this section resolves the limitations of previous measurement methods 
noted in \Cref{sec:pre_underestimate} and demonstrates the generality of our method.

\subsection{Inextractability Definition}\label{sec:method_definition}
Inspired by the differential privacy definition, which bounds the worst-case probability that an adversary can distinguish between neighboring datasets, we define \emph{inextractability} by bounding the worst-case extraction risk of the most extractable subsequence in the protected dataset.

\shortsection{Cost-bounding extraction}
Unlike inference games, extraction is a discoverable game where more attack attempts increase the probability of success~\cite{hayes2024measuring}.
In the extreme case, any text may be extracted as the number of independent extraction attempts approaches infinity, $n\to \infty$. We therefore couple extraction success with query cost for meaningful security. 
Following conventions in cryptography, $b$-bit security means that any attacker that successfully breaks the cryptographic primitive must incur a cost $n$ and achieve a success probability $p$ such that $\nicefrac{n}{p} \geq 2^b$~\cite{micciancio2018bit}.
Thus, we naturally derive \Cref{def:ext_bit}, where the protected set $D_\text{pro}$ in practice can be text collected from a sensitive domain (e.g., clinical notes) or copyrighted source (e.g., books).

\begin{definition}[\textbf{Inextractability}]\label{def:ext_bit}
An oracle $\mathcal{O}$ to an LLM satisfies \textbf{$(l, b)$-inextractability} with respect to a protected data set $D_\text{pro} \subseteq D$ 
if the expected number of independent trials required to extract any subsequence $\vz \in D_\text{pro}$ of at least $l$ grams via $\mathcal{O}$ is at least $2^b$. \end{definition}

\noindent
Similarly to the definition of security level in cryptography, we quantify the inextractability as the expected effort needed by a worst-case attack for a successful extraction.
The cost-based definition indicates that extraction is theoretically possible but requires at least $2^b$ trials; thus a smaller $b$ suggests a higher extraction risk, while a larger $b$ makes extraction economically infeasible.
It naturally captures the risk under multiple trials and provides a decoding-independent upper bound on extraction risk.
While Hayes et al.'s definition of $(n,p)$-discoverable~\cite{hayes2024measuring} also captures the multi-trial risk, it relies on a fixed decoding strategy.

To construct the $n$-trial worst-case extraction game, we assume each trial has the optimal strategy and derive the overall maximum extraction success in \Cref{lemma:independent_optimal} (proof in Appendix~\ref{sec:appendix_proofs}). %
We show how to estimate the optimal probability $p_{\vz}^*$ in \Cref{sec:method_worst}.

\begin{lemma}[Optimality of Independent Strategy]\label{lemma:independent_optimal}
Let $p_{\vz}^*$ be the highest single-trial extraction probability for any $l$-gram suffix $\vz_s$ in $\vz$. For any (possibly adaptive) $n$-trial strategy $\mathcal{A}_{n}$ the probability of extracting $\vz_s$ in $n$ attempts is at most
\begin{align}
\Pr\Bigg[\bigcup_{j=1}^n \{\mathcal{A}_j(\mathcal{O}, \vz_p)=\vz_s\}\Bigg] \leq p \,=\, 1 - (1 - p_{\vz}^*)^n,
\end{align}
where the equality is achievable by performing $n$ independent repetitions of an optimal single-attempt strategy.
\end{lemma}

\shortsection{Connection to min-entropy}
\Cref{lemma:independent_optimal} shows how extraction success grows with repeated trials, analogous with guessing a password with sufficient queries.
We relate this success growth to the attack cost underlying a min-entropy characterization, establishing a formal link between inextractability and information-theoretic security.
For an interface $\mathcal{O}$ satisfying $(l, b)$-inextractability, by definition, it requires at least $2^b$ expected attempts to extract any $l$-gram text. This translates to the constraint that $\forall \vz\in D_\text{pro}$:
\begin{align}
  2^b \leq \frac{n}{p} \quad \Rightarrow \quad b \leq \log_2 \left[ \frac{n}{1-(1-p_{\vz})^n} \right] \label{line:ieq_bits}
\end{align}

The right side of the inequality is a monotonically increasing function with respect to $n$. The attack success increases more slowly than the increase of cost, indicating that the most highest cost--success ratio is achieved by a single optimal trial. 
Therefore, taking the tightest bound over all $\vz\in D_\text{pro}$, we have:
\begin{align}
  b = \min_{\vz \in D_\text{pro}} (-\log_2 p_{\vz}) = -\log_2 \left(\max_{\vz \in D_\text{pro}} p_{\vz}\right)
\end{align}

This directly connects to the min-entropy of $D_\text{pro}$ under the optimal adversary strategy with $n$-trials $\mathcal{A}_n^*$:
\begin{align}
  H_{\infty}(D_\text{pro}|\mathcal{A}_n^*) = -\log_2 \left(\max_{\vz \in D_\text{pro}} \Pr[\mathcal{A}_n^*(\theta) = \vz]\right)
\end{align}

Intuitively, min-entropy measures how predictable the most extractable subsequence is. Analogous to evaluating the security level of a cryptographic system, \Cref{lemma:ext_entropy} establishes that bounding inextractability in~\Cref{def:ext_bit} is equivalent to ensuring a minimum entropy threshold:

\begin{lemma}[Equivalence of Inextractability and Min-Entropy]\label{lemma:ext_entropy}
If an interface $\mathcal{O}$ satisfies $(l, b)$-inextractability with respect to a protected set $D_\text{pro}$, it ensures the min-entropy $H_{\infty}(D_\text{pro}|\mathcal{A}) \geq b$ for any adversary $\mathcal{A}$.
\end{lemma}

\shortsection{Range of extraction cost}
To interpret the extraction risk at level 
$b$ in practice, we characterize its range from cases where an attacker can guess the secret in a single trial to cases where extracting a suffix from the model is no easier than sampling from a na\"{\i}ve baseline distribution.
When the model strongly memorizes a secret $\vz$, e.g., each ground truth token of it has the highest sampling probability, a greedy decoding results in the minimum extraction cost, $b=0$.

The maximum value of $b$ represents the min-entropy of $D_\text{pro}$ without the model access oracle, $H_\infty(D_\text{pro})$.
The gap $H_\infty(D_\text{pro})-b$ represents the information leakage~\cite{smith2009foundations} enabled by access to the target model.
Given  vocabulary size $|\mathcal{V}|$ and length $l$, the whole secret space is $|\mathcal{V}|^l$ and $H_\infty(D_\text{pro})\leq l\log_2(|\mathcal{V}|)$, where the equation holds with a uniform distribution over all possible sub-sequences.

Since the natural language space is highly non-uniform, a tighter estimation on $H_\infty(D_\text{pro})$ requires an assumption on the prior distribution of $D_\text{pro}$.
An unconditioned prior~\cite{swanberg2025beyond} assumes a token distribution over the vocabulary, e.g., by Zipf's law, and independently samples each token from this distribution.
A conditioned prior requires a general model trained without the protected data. %

We emphasize that these prior distributions rely on strong assumptions, and the obtained min-entropy only acts as a reference to select a cost budget $b$ that reflects the maximum reasonable attacker investment. 
In general, we have the valid range of $0\leq b\leq l\log_2(|\mathcal{V}|)$.

\begin{figure*}[htb]
\centering
\includegraphics[width=1.0\linewidth]{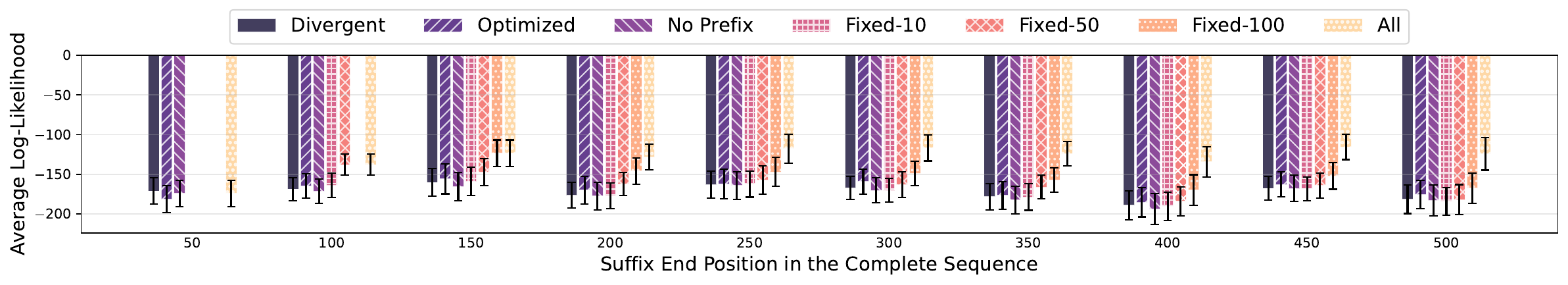}
\caption{Comparison of extraction prefixes. The log-likelihood is averaged over 200 extractions from the Enron dataset (suffix length $l=50$) using a fine-tuned GPT-2 with temperature $T=1$. Higher log-likelihood indicates a better prefix for extraction.
\textbf{All} (ours): uses all preceding tokens as the prefix for predicting the suffix starting at each position. 
\textbf{Divergent}~\cite{nasr2023scalable}: employs the divergent attack with their repeating-token prompt. 
\textbf{Optimized}~\cite{kassem2024alpaca}: initializes with the original 100-token prefix and iteratively optimizes it for 72 steps with a balanced objective that maximizes the longest common subsequence (LCS) while reducing overlap. 
\textbf{No Prefix}: uses the empty start-of-sequence token as prefix.
\textbf{Fixed-}$\mathbf{{|\vz_p|}}$: uses a fixed $|\vz_p|$-gram original prefix.
}
\label{fig:prefix}
\end{figure*}

\begin{figure}[htb]
    \centering
    \includegraphics[width=0.9\linewidth]{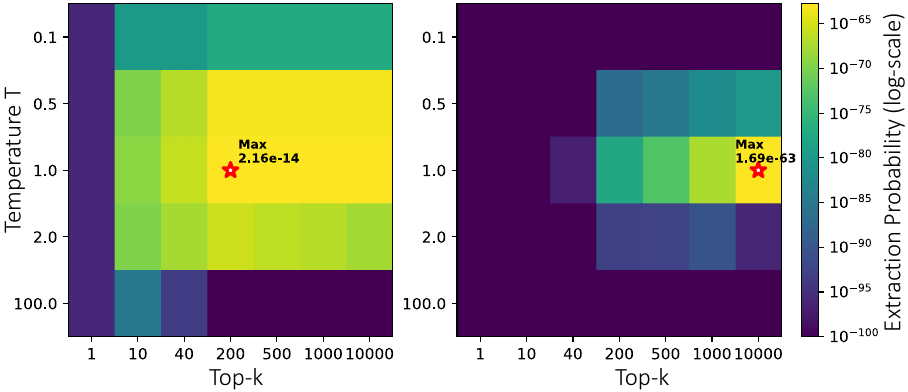}
    \caption{Extraction probability of fine-tuned GPT-2 with $l=50$, showing two training samples have different preferences on the decoding parameters, especially for the configuration such as top-$k$ that determines the truncation level.
    $T=1$ is the best as it aligns with the training objective.}
    \label{fig:subopt_decoding_kt}
\end{figure}

\begin{figure}[htb]
    \centering
    \includegraphics[width=0.95\linewidth]{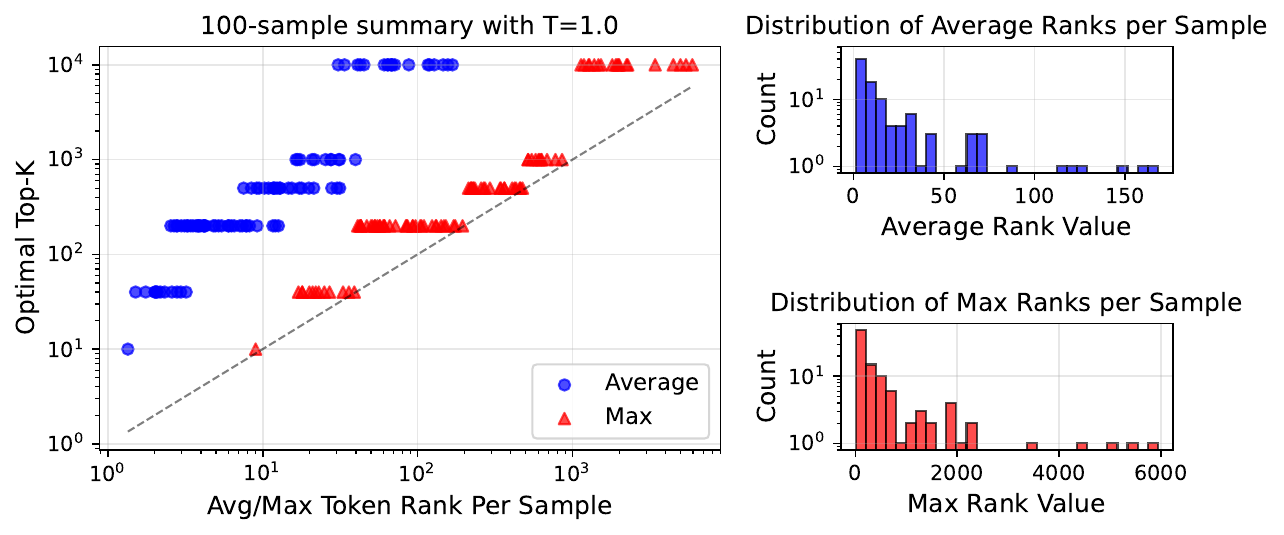}
    \caption{
    The optimal decoding parameter $k$ directly aligns with the sample’s maximum token rank. The large per-sample variations make fixed decoding configurations sub-optimal for extraction risk.
    }
    \label{fig:subopt_decoding_rank}
\end{figure}

\subsection{Worst-Case API Extraction Game} %
\label{sec:method_worst}
According to \Cref{lemma:independent_optimal}, to upper bound the worst-case extraction risk over $n$ trials, we need to construct the optimal attack that obtains an upper bound on $p_{\vz}$ for each trial.

\shortsection{Threat model for worst-case extraction}\label{sec:threat_model}
Following existing discoverable extraction measurements~\cite{hayes2024measuring,kassem2024alpaca}, we assume that the adversary has access to the complete text $\vz$.
The goal is to extract the suffix $\hat{\vz}_s=\vz_s$ from the model API oracle as $\hat{z}_s\leftarrow \mathcal{A}(\mathcal{O}_{f_\theta\circ\phi^*}, \vz_p^*)$, where the optimal decoding mechanism and prefix $\phi^*, \vz_p^*\leftarrow\mathcal{A}^\prime(\theta, \vz)$ are the output of another adversary procedure given the access to the model and the text.
We consider the setting in \Cref{tab:llm_access_control} where $\mathcal{O}$ releases the probabilities for the top-$m$ candidate tokens and $\phi$ has customizable decoding parameters $k$ and $T$.

\shortsection{Optimal prefix for verbatim extraction}
As discussed in \Cref{sec:pre_underestimate}, the standard fixed-length prefix is sub-optimal. To understand how much potential improvement is possible with an optimized prefix, we conducted experiments summarized in 
\Cref{fig:prefix}.
Even though the optimization consumes hundreds of times more queries to construct the optimized prefix, we find that it barely outperforms the original prefix with $|\vz_p| = 50$.
Across all of our experiments, we find that using all of the preceding tokens as the prefix provides the highest likelihood of extracting the suffix, although for most of our tests the increase over the 100-token fixed length prefix is small. An adversary with fixed resources would be more likely to succeed by using more queries with a 100-token fixed length prefix, than by incurring the additional cost of finding an optimal prefix or using a longer prefix.

\shortsection{Optimal decoding}
As illustrated in \Cref{fig:subopt_decoding_kt} and \Cref{fig:subopt_decoding_rank}, the decoding strategy substantially shifts the extraction probability $p_z$, and its effect is empirically non-monotonic across samples. 
A configuration that helps one extraction (e.g., smaller truncation to concentrate mass) may hurt another (e.g., excluding a mid-rank token).
Hence, using a fixed decoding systematically underestimates the worst-case risk.

Inspired by \Cref{fig:subopt_decoding_kt} and \Cref{fig:subopt_decoding_rank}, the optimal truncation size $k$ strictly aligns with the maximum rank of the ground truth token of $\vz_s$, which varies across samples and positions.
We therefore consider a stronger adaptive adversary who, for each extraction attempt at each position, optimizes truncation size $k$ and temperature $T$ to maximize the one-time success probability $p_z$. 
A realistic adversary would not be able to do this, of course, since they don't know the rank of the correct token, but our goal here is to simulate a worst-case adversary.
This strategy balances the two trade-offs: (1) smaller $k$ and lower $T$ concentrate probability mass on top-ranked tokens (exploitation); (2) larger $k$ or higher $T$ spread mass to lower-ranked ones (exploration). 

We obtain a worst-case upper bound for $p_z$ based on the rank-aware decoding optimization strategy in \Cref{thm:worst_case_extraction_risk}, with the formal proof in Appendix~\ref{app:worst_case_extraction_risk}.
The key idea is to choose $k$ equal to the token rank $r_t$ so that each ground-truth token is retained, and set $T$ sufficiently large to make the surviving tokens nearly uniform in probability. 
This yields an upper bound of $\nicefrac{1}{r_t}$ for each position. 

Aggregating across positions, the success probability of a single optimal attempt cannot exceed $\prod_t r_t^{-1}$, while $n$ independent attempts compound success via $1 - (1-\prod_t r_t^{-1})^n$. 
This gives a conservative, parameter-agnostic bound: tokens with smaller ranks $r_t$ are exponentially easier to extract, so increasing their rank sharply reduces overall extractability.
The bound treats each attempt independently and per-position ranks as fixed (teacher-forcing), so adaptive querying only compounds success through repetition.

\begin{theorem}[\textbf{Extraction Risk Upper Bound}]
\label{thm:worst_case_extraction_risk}
Given an LLM oracle $\mathcal{O}_{f_\theta\circ\phi_{k, T}}$ where user can control the decoding mechanism $\phi$ with parameters $k$ and $T$, and a training sample $\vz$, let $r_t$ denote the rank of the ground truth token $z_t$ in descending order of the probability distribution, $P_t(v)$ at position $t$ across every token $v\in\mathcal{V}$. 
The worst-case discoverable extraction risk over any truncation size $k$ and temperature $T$ is bounded by:
\begin{align*}
\max_{k, T} \Pr\Bigg[\bigcup_{j=1}^{n} \big\{ 
\mathcal{A}_j(\mathcal{O}_{f_\theta\circ\phi_{k, T}}, \vz_p) = \vz_s 
\big\}\Bigg] 
= 1 - \Big(1 - \prod_{t=1}^{l} \frac{1}{r_t}\Big)^n.
\end{align*}
\end{theorem}

\begin{algorithm}
\caption{Rank-Aware Estimation of Extraction Cost}\label{alg:one_shot}
\begin{algorithmic}[1]
\Require LLM $f_\theta$, available top logits size $m$, protected dataset $D_\text{pro}$, extraction length $l$
\Ensure \text{Inextractability level} $b$
\State Initialize the highest probability variable $p^* \leftarrow 0$
\For{each sequence $\vz \in D_\text{pro}$ of length $L$}
    \State Get ranks via teacher-forced pass $\{r_t\}_{t=1}^{L}\leftarrow f_\theta(\vz)$
    \For{$i = 1$ to $L - l + 1$}
        \State $p_{\vz} \leftarrow \prod_{t=i}^{i+l-1} p_t$, $p_t = \begin{cases} 1/r_t & r_t \le m \\ P_t(z_t) & \text{otherwise} \end{cases}$
        \If{$p_{\vz} > p^*$}
            \State $p^* \leftarrow p_{\vz}$
        \EndIf
    \EndFor
\EndFor
\State \Return $b \leftarrow -\log_2(p^*)$
\end{algorithmic}
\end{algorithm}

\shortsection{Rank-Aware Estimation Framework}
Given the rank-based extraction risk bound, we propose a rank-aware estimation in \Cref{alg:one_shot}.
Different from generation-based extraction measurements~\cite{carlini2021extracting,carlini2022quantifying} that rely on generation output to compute the extraction ratio, we follow the probabilistic risk estimation~\cite{hayes2024measuring} which uses the product of sampling probabilities at all positions to estimate the one-time extraction probability, $p_z$.
Thus, the estimation framework only involves one-shot forward pass, which is efficient and the resulting probability matches the empirical extraction ratio~\cite{hayes2024measuring}.
The teacher-forcing forward naturally selects the empirically optimal prefix with all preceding tokens before each position.
Additionally, the sliding window design enables an efficient scan over overlapping subsequences in a protected sequence.

Instead of using a fixed decoding mechanism $\phi$~\cite{hayes2024measuring}, our framework estimates the upper bound of $p_z$ especially when setting the worst-case $m=|\mathcal{V}|$, ensuring that no ground truth token is truncated.
This is a practical upper envelope: when $r_t\le m$, $\nicefrac{1}{r_t}$ upper bounds any achievable probability under adaptive $(k, T)$; when $r_t>m$, the exact softmax value is applied, which serves as a practical estimate below the unattainable $\nicefrac{1}{r_t}$.

\subsection{Connection to Other Measurements}\label{sec:method_connection}
Generation-based and probabilistic extraction are representative approaches for targeted and exact measurement.

\shortsection{%
Greedy generation}
Generation-based extraction~\cite{carlini2022quantifying,slack2025early} prompts the model with the prefix and uses greedy decoding to estimate the ratio of extractable suffixes in $D_\text{pro}$.
In practice, two evaluation protocols are used: \textit{chunked} evaluation splits $D_\text{pro}$ into $l$-length non-overlapping prefix–suffix pairs, which is efficient but may miss extractable sequences that span chunk boundaries; \textit{sliding window} evaluation scans all consecutive prefix–suffix pairs over every token offset, avoiding missed pairs at the cost of  $(L-l+1)\cdot l$ model forward passes, where $L$ is the sequence length. 

Since greedy decoding is deterministic, we propose \Cref{alg:for_greedy} to tightly estimate the extractable ratio based on our rank-aware estimator.
Unlike chunked evaluation, it scans all token offsets and thus misses no extractable prefix--suffix pair. 
Unlike sliding window evaluation, it requires only a single teacher-forced forward pass per sequence, reducing the cost from $(L-l+1)\cdot l$ forward passes to one.
As long as the extraction rate with greedy decoding is not zero, the inextractability in \Cref{def:ext_bit} hits the lowest $b=0$, indicating there exists at least one $l$-gram sample can be deterministically extracted.

\begin{algorithm}[tb]
 \caption{Efficient Estimation for Greedy Generation}
 \label{alg:for_greedy}
 \begin{algorithmic}[1]
 \Require LLM $f_\theta$, protected data $D_{\mathrm{pro}}$, extraction length $l$
 \Ensure Extractable rate $\eta$
 \State $c_\text{tot}, c_\text{ext} \leftarrow 0, 0$
 \For{each sequence $\vz \in D_\text{pro}$ of length $L \ge l$}
  \State Get ranks via teacher-forced pass $\{r_t\}_{t=1}^L\leftarrow f_\theta(\vz)$
  \State $i \leftarrow 1$ \Comment{Sliding window start pointer}
  \While{$i \le L - l + 1$}
    \State Find first position $t^* \in [i, i+l-1]$ where $r_{t^*} > 1$
    \If{no such $t^*$ exists}
      \State $c_\text{ext} \leftarrow c_\text{ext} + 1$
      \State $i \leftarrow i + 1$
    \Else
      \State $i \leftarrow t^* + 1$ \Comment{Skip windows that would fail}
    \EndIf
  \EndWhile
  \State $c_\text{tot} \leftarrow c_\text{tot} + (L-l+1)$
 \EndFor
 \State \Return Extraction rate $\eta\leftarrow c_\text{ext} / c_\text{tot}$
 \end{algorithmic}
\end{algorithm}

\shortsection{Conversion to probabilistic measurement}
The recent $(n, p)$-probabilistic extraction~\cite{hayes2024measuring} defines a sample as \emph{extractable} when its generation probability within $n$ trials given a certain decoding configuration is greater than $p$.
They compute the required trials with $n\geq \log{(1-p)}/\log{(1-p_z)}$, where $p_z$ is generation probability of sample $z$ given a specific configuration ($k=40$, $T=1$).

While our \Cref{def:ext_bit} focuses on one deterministic success, we derive the exact conversion of required query cost in our definition of \emph{$(l, b, \delta)$-inextractability} in \Cref{def:bit_delta} given the probabilistic success rate $p=1-\delta$.
When the maximum one-shot probability $p_{\vz}^*$ is small, the cost increases by approximately a $\ln(1/\delta)$ factor, which is derived from the Taylor expansion of $-\ln(1-p_{\vz}^*)=p_{\vz}^*+\tfrac{p_{\vz}^{*2}}{2}+O(p_{\vz}^{*3})$ for $p_{\vz}^*\ll 1$.
Compared to the exact number of required trails $n_\text{exact}=\frac{\ln(1/\delta)}{-\ln(1-p_{\vz}^*)}$, the approximated $n_\text{stable}=\ln(1/\delta)/p_{\vz}^*$ is numerically more stable.

\begin{definition}[\textbf{Probabilistic Extraction Risk}]\label{def:bit_delta}
An oracle $\mathcal{O}$ to an LLM satisfies \textbf{$(l, b, \delta)$-inextractability} with respect to a protected set $D_\text{pro} \subseteq D$ if it requires at least $2^{b(\delta)}$ expected independent trials to extract any subsequence $\vz\in D_\text{pro}$ of length at least $l$ grams with probability at least $1-\delta$, where $2^{b(\delta)} = \frac{\ln(1/\delta)}{-\ln(1-2^{-b})} \approx \ln(1/\delta)\cdot 2^b$ for small $2^{-b}$.
\end{definition}

\begin{figure*}[htb]
\centering
\includegraphics[width=0.48\linewidth]{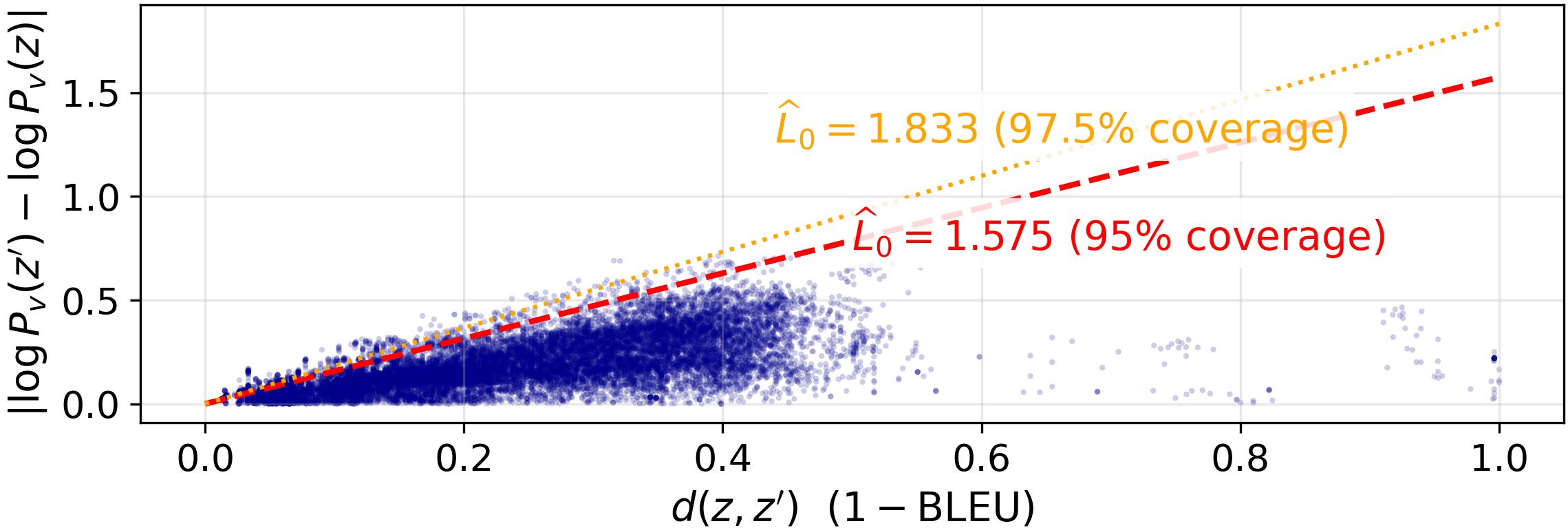}
\includegraphics[width=0.49\linewidth]{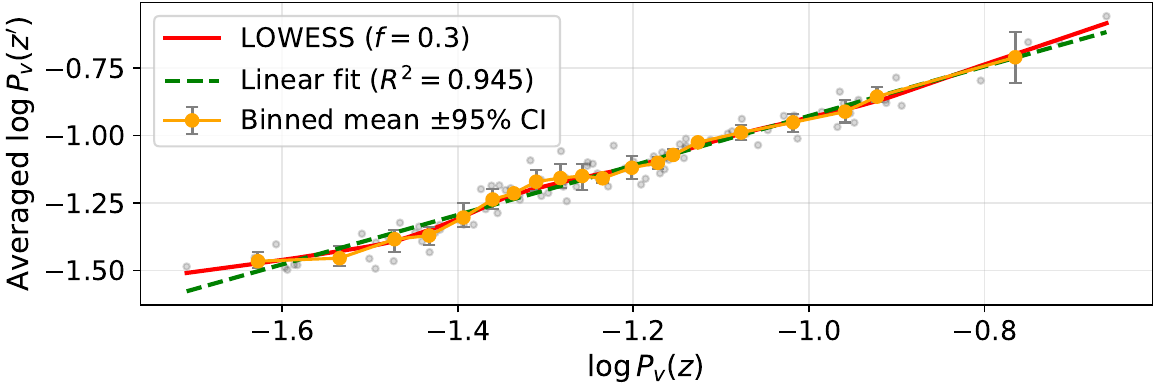}
\caption{
\textbf{(Left)}: Correlation between the per-pair log-probability difference and the distance for center-neighbor pairs. 
The dashed and dotted lines show estimated local log-Lipschitz bounds covering 95\% and 97.5\% of pairs.
\textbf{(Right)}: Correlation between the length-normalized verbatim extraction probability $P_v(z)$ (x-axis) and the average $P_v(z^\prime)$ (y-axis) for neighbors with $d(z, z^\prime)<0.2$. 
The monotonic trend suggests that reducing the extractability of a center text also shrinks the average extraction probability of its neighbors.
}
\label{fig:beyond_prob_change}
\end{figure*}

While our \Cref{def:bit_delta} has an equivalent form as $(n, p)$-probabilistic measurement, our estimate tightly upper bounds the probabilistic measurement.
Because we consider the maximum one-shot generation probabilistic $p_z^*$ over any decoding configurations, while in $(n, p)$-probabilistic extraction that value of $n$ is derived given the measured probability $p_z$ with a specific $\phi$.
Additionally, we take $(l,b)$ as the core guarantee and treat $\delta$ as an external conversion knob because of:
(1) \emph{simplicity}: one scalar simplifies comparison and is clearer for governance than a pair of scalars; 
(2) \emph{interpretability}: $2^b$ essentially measures the attacker's investment/success ratio $n/p$, which is maximized at $n=1$; and
(3) \emph{conservatism}: dropping $\delta$ in the core metric prevents overfitting guarantees to a single confidence target, and probabilistic forms just multiply by $\ln(1/\delta)$.

\subsection{Untargeted and Approximate Extraction}\label{sec:method_beyond}
Our risk measurement is defined with the exact match criterion and the targeted attack goal.
While it mainly reflects the targeted verbatim generation risk, we discuss its extension to other data extraction goals here.

\shortsection{Untargeted extraction} We extend the targeted $(l,b)$-inextractability guarantee to an untargeted setting, where success corresponds to emitting any protected $l$-gram. 
The following theorem (proven in Appendix~\ref{sec:appendix_proofs}) 
quantifies the induced extractability erosion:
\begin{theorem}[Untargeted Inextractability]\label{thm:targeted_to_untargeted}
Suppose the LLM oracle $\mathcal{O}$ satisfies $(l,b)$-inextractability with respect to $D_{\text{pro}}$ (Definition~\ref{def:ext_bit}), with $\mathcal{S}=\{\vz\in D_{\text{pro}}:|\vz|=l\}$ as the set of protected $l$-grams and $|\mathcal{S}| = M$. Then the single-attempt probability that $\mathcal{O}$ emits some protected $l$-gram satisfies
$
\Pr[\exists \vz\in\mathcal{S}: \mathcal{O} \text{ emits }\vz] \le \min\{1, M \cdot 2^{-b}\}.
$ \\
Defining the untargeted inextractability
$$
b_{\mathrm{un}} := -\log_2 \Pr[\exists \vz\in\mathcal{S}: 
\mathcal{O} \text{ emits }\vz],
$$
we obtain the lower bound
$
b_{\mathrm{un}} \ge \max\{0,\, b - \log_2 M\}.
$
\end{theorem}

If the per-$\vz$ emission events are additionally assumed to be conditionally independent,
the union bound tightens to $b_{\mathrm{un}} \ge -\log_2\big(1-(1-2^{-b})^{M}\big)$, which strictly improves the union-bound expression when $M \cdot 2^{-b}\ll 1$. 
In either case, moving from targeted to untargeted extraction incurs a loss of at most $\log_2 M$-inextractability; choosing $b \ge \log_2 M + \log_2(1/\gamma)$ ensures that the probability of emitting any protected $l$-gram in one trial is at most $\gamma$.

\shortsection{Approximate matching} 
Exact (verbatim) bounds are more conservative than approximate-match measurements~\cite{ippolito2022preventing,hayes2024measuring,kassem2024alpaca,barbero2025extracting}. 
Directly applying \Cref{thm:targeted_to_untargeted} to all approximate variants is infeasible because the set of paraphrases and small edits is practically unbounded. We therefore ask: \emph{under what conditions does reducing the extraction risk of a center text also reduce that of its neighborhood?}

Let $d(\cdot,\cdot)$ be any nonnegative distance or dissimilarity measure over sequences with $d(\vz,\vz)=0$ (e.g., edit distance, embedding-based similarity, $1-$BLEU and $1-$ROUGE-L 
or other task-specific similarity thresholds). 
For a reference sample $\vz$ and radius $c>0$ define the closed neighborhood
$
\mathcal{N}(\vz) := \{ z' : d(\vz,\vz') \le c \}.
$

Since $\mathcal{N}(\vz)$ may contain sequences of different lengths, let $P_v(\cdot)$ denote the length-normalized version of $p_{\vz}^*$.
To avoid numerical instability on very small probabilities, we work in log space. 
The following local regularity assumption is stated for a general metric $d$; $L_0$ can be empirically estimated as shown in \Cref{fig:beyond_prob_change} (estimation for other distances in Appendix~\ref{APP:sec:add_results}).

\begin{assumption}[Local Log-Lipschitz Regularity]\label{ass:lipschitz}
There exists $L_0\ge 0$ such that for all texts $x,y\in\mathcal{N}(\vz)$,
$$
\big|\log P_v(x) - \log P_v(y)\big| \le L_0\, d(x,y).
$$
\end{assumption}

\noindent Based on \Cref{ass:lipschitz}, \Cref{cor:ratio_suppression} shows defenders can target an inextractability threshold: a defense improving the inextractability level by $\Delta_b > (L_0+L_0')c/\ln 2$ suppresses the average likelihood of generating its approximately similar neighbors.
We illustrate this correlated trend in \Cref{fig:beyond_prob_change} with BLEU based distance and $c=0.2$.

\begin{corollary}[Ratio-Based Neighborhood Suppression]\label{cor:ratio_suppression}
Assume pre- and post-defense log-Lipschitz regularity with constants $L_0$ and $L_0'$ over radius $c$ (Assumption~\ref{ass:lipschitz}). If the defense improves the verbatim inextractability by $\Delta_b$ bits
$$
\Delta_b > \frac{(L_0+L_0')c}{\ln 2},
$$
then the average worst-case verbatim generation probabilities for all neighbors in $\mathcal{N}(z)$ are jointly suppressed.
\end{corollary}

\section{Evaluation of Extraction Risk}\label{sec:experiment}

This section reports on experiments to evaluate the defined inextractability and analyze its properties. 
\Cref{sec:mia_ext} analyzes the interplay between extractability and indistinguishability;
\Cref{sec:comparing} compares different extraction measurement methods; \Cref{sec:understandingcost} employs our methods to understand extraction cost.

Instead of using real-world LLM APIs, we use open-sourced model for a fully controlled $m$.
Following previous extraction works~\cite{carlini2022quantifying,hayes2024measuring,kassem2024alpaca}, we use GPT-2-Small~\cite{radford2019language} on the Enron email dataset~\cite{klimt2004enron} to evaluate risk of fine-tuning data, and use LlaMA-3.1-8B~\cite{grattafiori2024llama} on the Pile~\cite{gao2020pile} and BookSum~\cite{kryscinski2021booksum} to evaluate pre-training data extraction.

\subsection{Extractability and Distinguishability}\label{sec:mia_ext} %

Comparing various MIAs with different signals and our estimated extraction cost, we obtain results similar to \Cref{fig:mia_ext_sample} and include one example with $l=50$ in \Cref{APP:fig:mia_ext} (in the Appendix).
\Cref{fig:mia_ext_corr} summarizes the correlation coefficient between the different membership inference attack signals and extraction cost. 

\shortsection{Influence of extraction sequence length}
For all MIA signals, a consistent observation is that the correlation between indistinguishability and extractability increases with longer sequences with larger $l$.
This is because shorter member or non-member samples are more ambiguous than longer sequences~\cite{duan2024membership}, blurring the correlation between extractability and distinguishability.

\shortsection{Influence of membership signal}
Consistent with our discussion in \Cref{ssec:explaininginconsistencies}, the calibration-based MIA signals (Zlib, Ref-Loss) have less correlation with the extraction costs than uncalibrated  MIA signals (Loss, Mink).
From the AUC of different MIAs (\Cref{APP:fig:auc_correlation} in Appendix), we reveal that the correlation between distinguishability and extractability depends more on whether the MIA signal is uncalibrated than on its ability to distinguish members.

\begin{figure}
\centering
\includegraphics[width=0.95\linewidth]{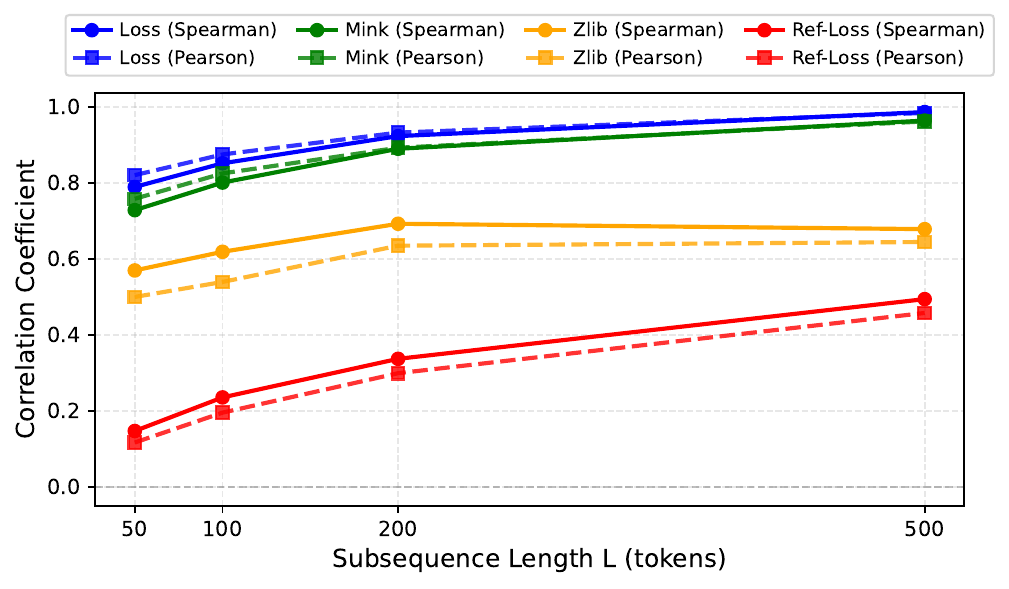}
\caption{Correlation between indistinguishability (1-MIA AUC) and Indistinguishability (Extraction Cost).}
\label{fig:mia_ext_corr}
\end{figure}

\begin{figure}
\centering
\includegraphics[trim=0 0 0 0, clip, width=0.92\linewidth]{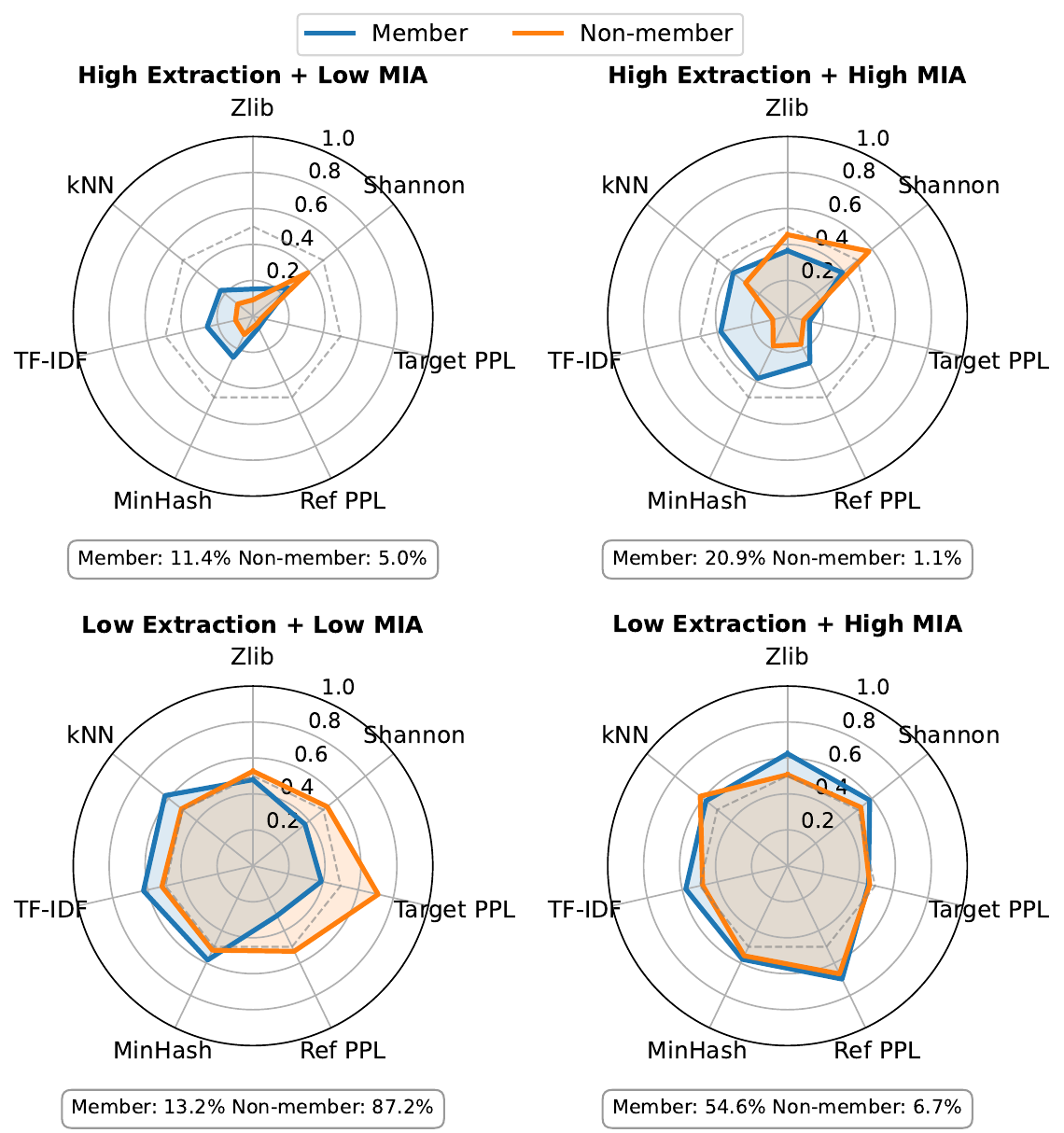}
\caption{
Radar charts comparing text properties (Zlib, Shannon entropy, Target/Ref PPL, MinHash, TF-IDF) between member (blue) and non-member (orange) samples, on a fine-tuned GPT-2 model (Enron, $l=500$, $m=20$). 
Axes show mean empirical quantiles (0-1) within each group, stratified by high extraction risk (cost$<1000$) and high MIA risk (confidence$>0.95$).
Groups percentages indicate the fraction of all member / non-member samples falling into each group.
}
\label{fig:mia_ext_radar}
\end{figure}

\shortsection{Influence of text properties}
We analyze the underlying properties for a text subsequence to have high distinguishability risk or extractability risk from the following aspects:
1) \textit{complexity}, including structural complexity Zlib$\uparrow$ and lexical complexity TF-IDF$\uparrow$; 
2) \textit{diversity} (Shannon entropy$\uparrow$);
3) \textit{uniqueness}, including \emph{n}-gram level (MinHash-based Jaccard Similarity$\downarrow$), and semantic level (kNN density$\downarrow$).
We use the best-performing reference-model~\cite{carlini2022membership} signal to indicate the per-sample MIA risk, and use our measurement to indicate the per-sample extraction risk.
For a transparent data split in MIA, we demonstrate with GPT-2 fine-tuned with the Enron dataset.
To train the reference model, we use an auxiliary set disjoint from fine-tuning data.

\shortsection{Factors influencing extractability} The radar plots in \Cref{fig:mia_ext_radar} reveal two distinct vulnerability modes. 
(1) \emph{Low-complexity / low-diversity} texts: with low structural (Zlib) or lexical (TF--IDF) complexity and low diversity (Shannon), such as form emails and boilerplate notices; their repetitive structure concentrates probability mass and enables deterministic regeneration. 
(2) \emph{High-uniqueness} texts: high \emph{n}-gram distinctness (low MinHash similarity) and sparse semantic neighborhoods (low kNN density), whose continuations lie in a narrow token set, making exact suffix emission likely once the prefix is known. 
We show in Appendix~\ref{APP:sec:add_results} that the two types are neither identical nor mutually exclusive with 0.3--0.5 Spearman correlation.

\vspace{1ex}
\noindent\emph{Implication.} Na\"{i}ve deduplication can backfire---aggressively removing near-duplicates may elevate relative \emph{n}-gram or semantic uniqueness and increase worst-case extractability.

\shortsection{Factors influencing distinguishability} In contrast, \Cref{fig:mia_ext_radar} indicates membership inference risk is not strongly associated with structural or lexical complexity; the dominant driver is the per-sample perplexity improvement from the reference model to the target model. The ``High MIA'' group occupies a larger aggregate radar area without extreme simplicity, indicating that extraction and membership inference target partially disjoint victim sets. 

\vspace{1ex}
\noindent\emph{Implication.} Audits limited to MIAs may miss samples that are at highest extraction risk under a worst-case attack.

\begin{figure*}
\centering
\includegraphics[width=\linewidth]{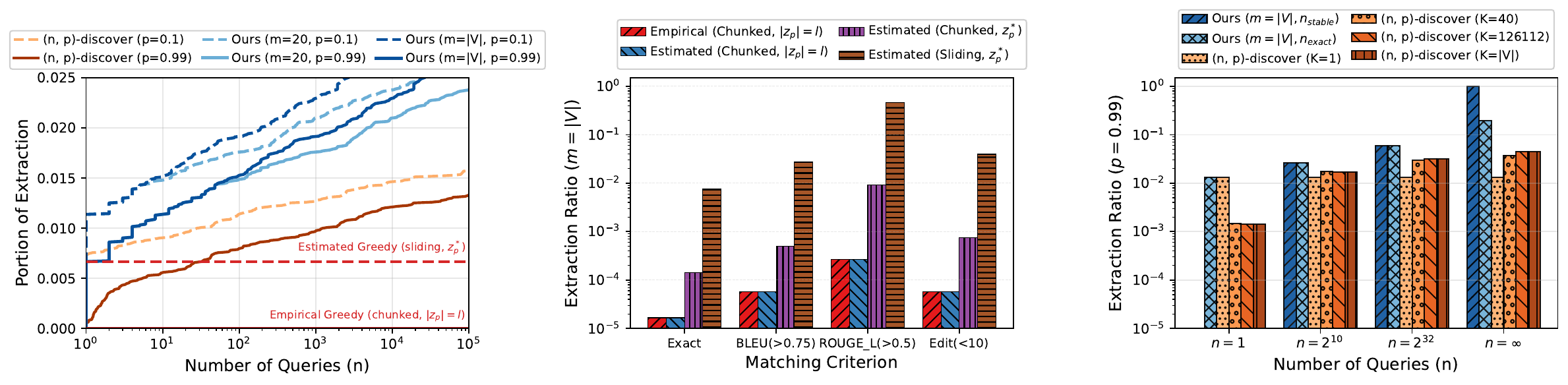}
\caption{Comparison with generation-based~\cite{carlini2022quantifying} and probabilistic extraction~\cite{hayes2024measuring} on Llama-3.1-8B with Pile-CC. 
The extraction portion denotes the ratio of extractable sub-sequences identified across $9\times10^5$ sliding windows of length $l=50$.
\textbf{(Left)} is exact extraction comparison. 
\textbf{(Middle)}: the tight estimation of \Cref{alg:for_greedy} on greedy generation across different matching metrics with $n=1$, and our advantage of using all prefix and comprehensive scanning via efficient sliding window.
\textbf{(Right)}: compares probabilistic extraction with  $\delta=0.01$ across truncation size $k$; 
$k=126122$ is the maximum true token rank $r_t$ over all samples;
$n_\text{exact}$ and $n_\text{stable}$ denote the exact and approximated numbers for required trials to achieve success rate of 1-$\delta$ in \Cref{def:bit_delta}.
}
\label{fig:extraction}
\end{figure*}

\subsection{Comparing Extraction Risk Measurements}\label{sec:comparing}

We compare our measurement with both generation-based and probability-based measurements in \Cref{fig:extraction}.

\shortsection{Limitations of empirical generation}
In \Cref{fig:extraction} (Left), the extraction ratio with greedy decoding does not scale with $n$ because it is deterministic.
And the extraction ratio is low for two reasons: 
(1) the limited prefix length $|\vz_p|$ is less informative than using the full preceding context; and
(2) for the chunked version, the chunking procedure  misses vulnerable samples across chunk boundaries.

\shortsection{Limitations of probabilistic extraction}
While the extraction ratio of $(n,p)$-discovery  in \Cref{fig:extraction} (Left) scales with $n$, 
(1) it is inferior to greedy generation with small $n$ and high success probability threshold $p=0.99$ because the specific decoding with $k=40$ dilutes the rank-1 confident tokens; 
(2) from \Cref{fig:extraction} (Right), the optimal decoding configuration influences the extraction ratio and the optimal choice is inconsistent across $n$; 
and (3) the extraction ratio is capped at around $0.015$ even with a very large $n=10^5$ because samples whose ground-truth token rank exceeds the truncation size $k$ have their sampling probability truncated to zero and can never be extracted.

\shortsection{Our advantage}
Our \Cref{alg:one_shot} provides a better overall upper bound than other measurements, as a result of using all prefix and optimal decoding.
In \Cref{fig:extraction} (Middle), our \Cref{alg:for_greedy} tightly estimates the greedy generation on chunked text with fixed prefix.
Using the full prefix allows us to find more vulnerable samples with $n=1$.
Further, the efficiency of \Cref{alg:for_greedy} allows us to scan all sliding window text, avoiding missing any vulnerable sequences.
As seen in  \Cref{fig:extraction} (Right), our \Cref{alg:one_shot} provides a consistent risk upper bound without relying on a specific decoding mechanism $\phi$.
In practice, $n_\text{stable}$ avoids underestimating the risk because it captures tail samples with $p^*_{\vz}$ approaching 0.

\subsection{Understanding Extraction Cost}\label{sec:understandingcost}

Our \Cref{def:ext_bit} assumes a strong adversary able to control the prefix and decoding strategy to maximize suffix recovery. Because real attackers cannot achieve this idealized power, we analyze the gap between the theoretical cost $2^b$ (with $1\le b\le l\log_2 |\mathcal{V}|$) and realistic cost.

\begin{figure}
\centering
\includegraphics[trim=0 0 0 21, clip, width=0.8\linewidth]{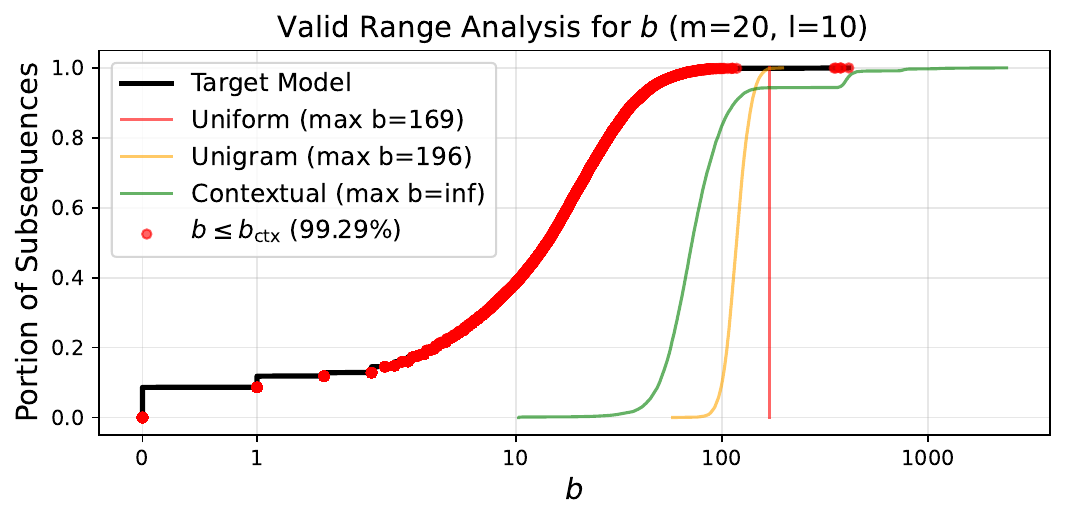}
\caption{
Analysis of valid $b$ range on extracting Pile-CC ($m=20, l=10$) from Llama-3.1-8B.
\textit{Blind baselines:} Uniform, Unigram, and Contextual have progressively stronger priors. 
\textit{In-context baseline:} assumes the target suffix appears in the prefix, providing a threshold $b_{\text{ctx}}$. 
Subsequences with cost below $b_{\text{ctx}}$ are highly vulnerable, while those exceeding the maximum blind-baseline cost are comparatively less vulnerable.}
\label{fig:rangeb_blind}
\end{figure}

\shortsection{Maximum cost from blind tests}
\Cref{fig:rangeb_blind} shows three blind baselines approximating the minimum extraction cost without access to the target model: 
(1) \emph{Uniform}, tokens sampled uniformly;
(2) \emph{Unigram}; sampled from empirical token frequencies; and 
(3) \emph{Contextual}, proxy LM conditioned on preceding tokens.
Stronger priors progressively lower cost. 
For many subsequences, access to the target model collapses the minimum $b$ from around 10 (contextual blind) to near $0$, highlighting vulnerable samples. 
Blind baselines thus provide conservative (contextual) and optimistic (unigram) resilience thresholds.

\shortsection{Minimum cost from in-context test}
In \Cref{fig:rangeb_blind}, we define an in-context baseline where the target suffix is fully revealed in the prompt. 
Specifically, given a complete text $z$, we prompt the target model as follows:
\begin{tcolorbox}[
    colback=gray!8,
    colframe=gray!40,
    boxrule=0.4pt,
    arc=2pt,
    left=4pt,
    right=4pt,
    top=3pt,
    bottom=3pt,
]
\small
\textbf{Prompt:} \textit{Please repeat after me:} "Mine ear is open, and my heart prepared". Mine ear is open, and my \\
\textbf{→ Expected:} heart prepared
\end{tcolorbox}

\noindent
The optimal decoding cost $\phi^*$ under this repeat prompt sets a practical lower bound $b_{\text{ctx}}$ on realistic leakage. 
If a subsequence's cost falls below $b_{\text{ctx}}$, it is flagged as high-risk. 
Over $99\%$ of protected subsequences remain below the in-context baseline even with repeated prompts, indicating extraction of content from training data rather than explicit context leakage.

\begin{figure}[tb]
\centering
\includegraphics[trim=0 0 0 20, clip, width=0.75\linewidth]{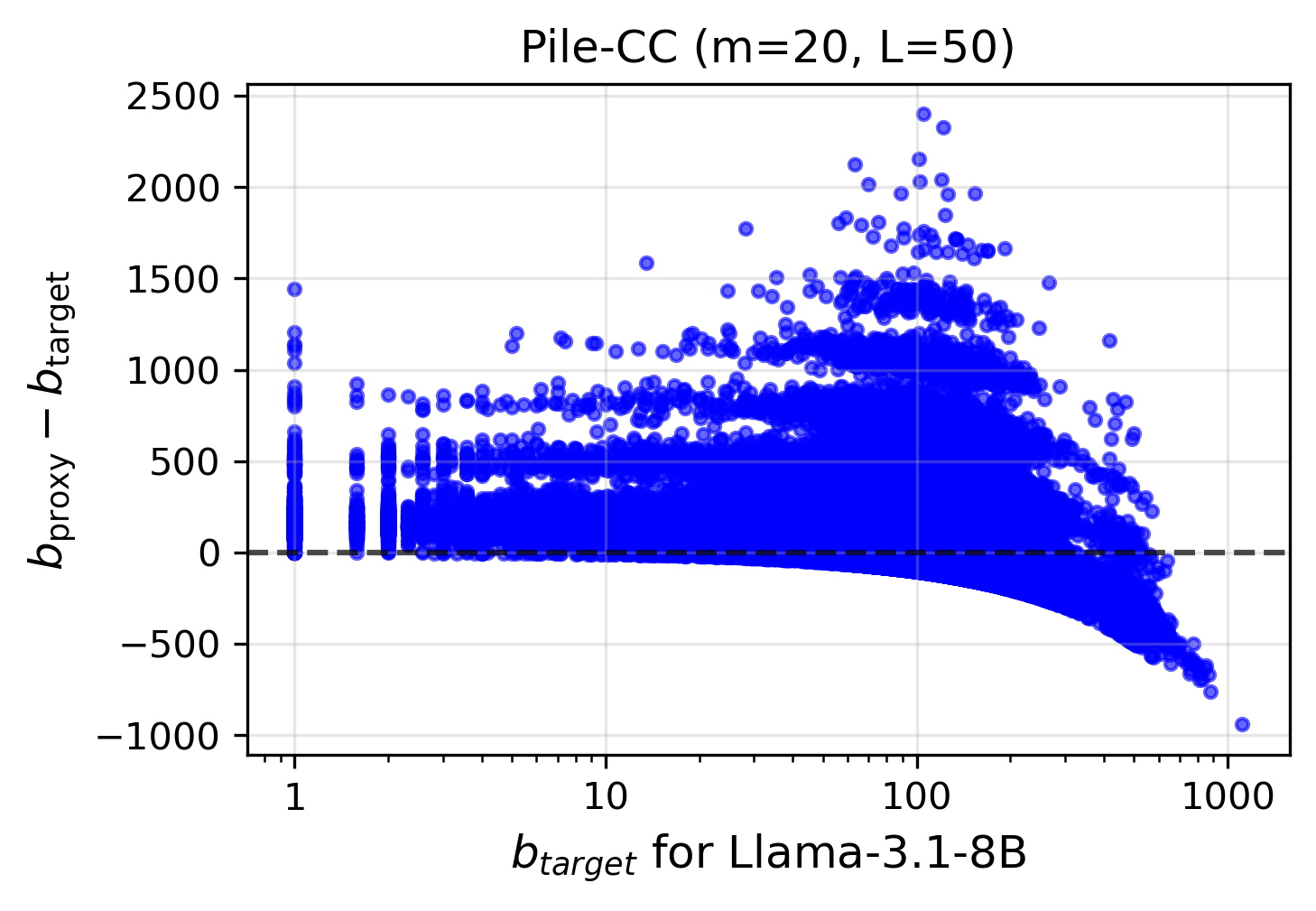}
    \caption{Privacy cost comparison on the Pile-CC ($m=20, L=50$). Each point represents one subsequence. 
    Across 483,485 subsequences, the target exceeds the proxy model (GPT-2) in 398,385 cases (82.4\%).
    }
    \label{fig:rangeb_reduction}
\end{figure}

\begin{figure*}
\centering
\includegraphics[width=0.48\linewidth]{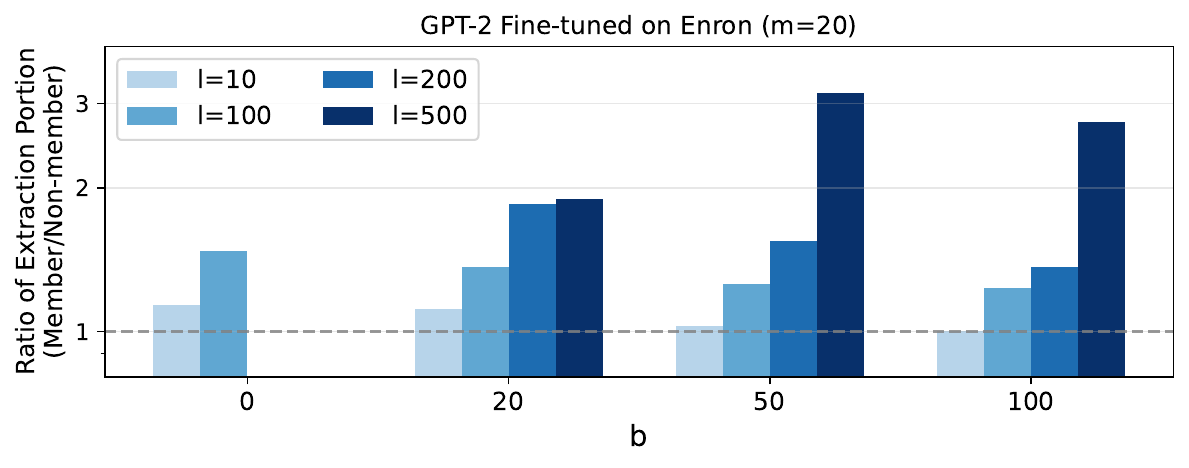}
\includegraphics[width=0.48\linewidth]{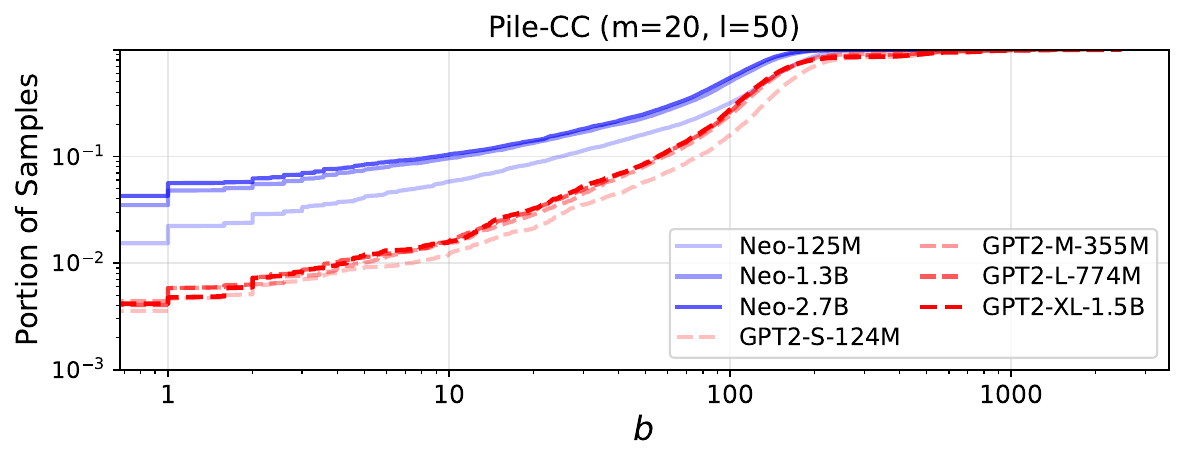}
\caption{
Extraction difficulty of member and non-member data across sequence length $l$ and model scale.
\textbf{(Left)}: Ratio of extraction portions (member/non-member) in Enron at fixed $b$.
\textbf{(Right)}: Following~\cite{carlini2022quantifying}, the Pile-CC dataset (not trained by GPT-2) is treated as the non-member split for GPT-2.
}
\label{fig:mem_non_scale}
\end{figure*}

\begin{figure}
\centering
\includegraphics[width=0.9\linewidth]{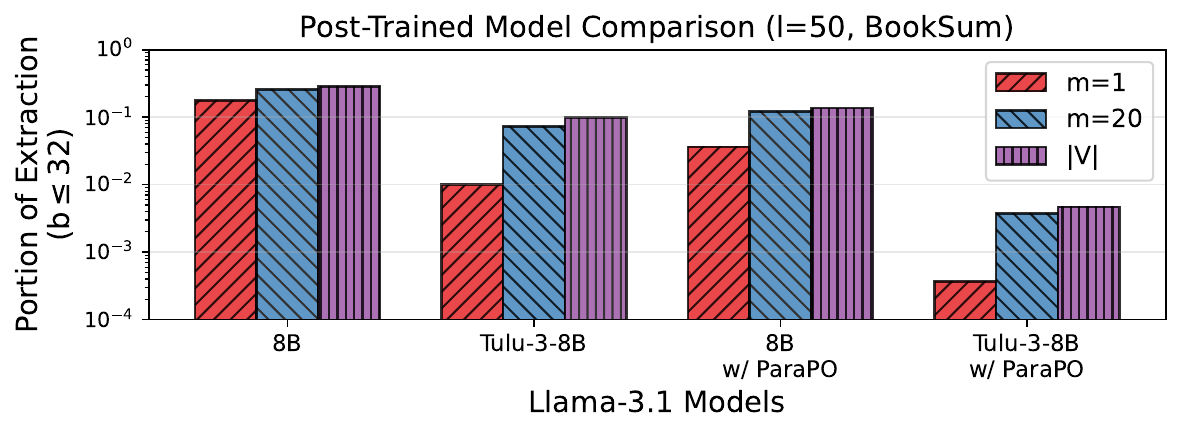}
\includegraphics[width=0.9\linewidth]{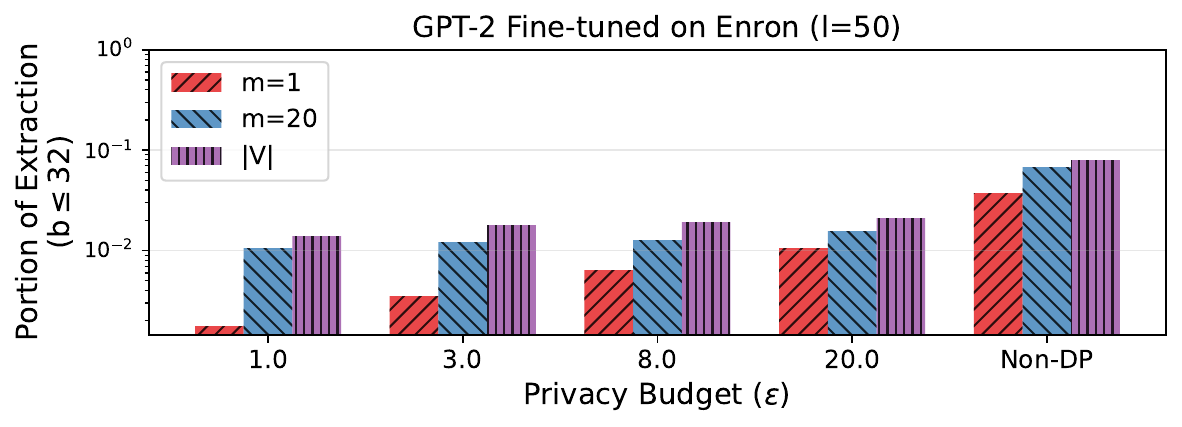}
\caption{Defense Evaluation.
\textbf{(Up)} combining instruction-tuning (Tulu-3) with preference optimization (ParaPO) substantially mitigates risks.
\textbf{(Bottom)} DP training reduces extractability but shows limited gains under small $\epsilon$, suggesting a gap between distinguishability defense and extraction.}
\label{fig:defense}
\end{figure}

\shortsection{Calibrated extraction cost}
To filter out common memorization, prior works introduce the notion of counterfactual memorization~\cite{zhang2023counterfactual,pappu2024measuring}, which calibrates memorization scores using a proxy model that has not seen the target sample.  
As shown in \Cref{fig:rangeb_reduction}, samples with low extraction cost tend to exhibit non-negative counterfactual memorization, indicating that our extraction cost can effectively identify uncommon memorization.

\shortsection{Member vs.\ non-member} 
In \Cref{fig:mem_non_scale}, we compare the extractability of training (member) and non-training (non-member) for fine-tuned (Left) and pre-trained (Right) models.
We observe that members consistently exhibit higher extraction portion than non-members given $b$.
For the fine-tuned GPT-2 on Enron, we compare by the ratio of extraction portions to quantify the memorization gain from being included in training.
A larger member-to-non-member ratio indicates that member samples are substantially easier to extract.
The gap is generally small, because members and non-members are randomly split and the model generalizes well.
For both fine-tuned and pre-trained models, there exist non-member samples that can be extracted with low cost, implying that verbatim generation does not require direct inclusion in training data~\cite{liu2025language}.

\shortsection{Influence of sequence length} 
Longer sequences (larger $l$) naturally correspond to higher extraction costs, as probability of extraction decays with $l$.
In \Cref{fig:mem_non_scale} (Left), the ratio increases with larger $l$, for two reasons:
(1) short sequences of members and non-members overlap more easily, making attribution ambiguous~\cite{duan2024membership};
(2) short suffixes even overlap with their prefixes, so they are likely to be retrieved than generated.
Hence, in practice, $l$ should be chosen large enough to avoid overlap between the prefix and the suffix.

\shortsection{Influence of model size}
As shown in \Cref{fig:mem_non_scale} (Right), larger models (darker lines) tend to have lower extraction costs than smaller models.
The member–non-member gap is also more pronounced for pre-trained models, possibly due to the unknown training corpora differences.

\section{Defenses for Extractability}\label{sec:defense}

This section briefly discusses possible mitigations for reducing extractability. 

\shortsection{Training-Phase Defenses}
Most previous works consider the extraction boundary as the white-box surface and thus focus on training-phase defenses, including
(1) \emph{data-processing}, including  deduplication~\cite{lee2021deduplicating} or sanitization~\cite{morris2022unsupervised,staab2024large};
(2) \emph{privacy-preserving training}, including incorporating DP noise in optimization~\cite{abadi2016deep} and empirical defenses based on regularization~\cite{hans2024like,tran2025tokens}; and
(3) \emph{post-training}, which introduces extra training stage for mitigating extractability risk of the previous model, such as ParaPO~\cite{chen2025parapo}. Instruction-tuned models~\cite{lambert2024tulu} have also been found to reduce verbatim generation~\cite{chen2025parapo}.

\shortsection{Inference-Phase Defense}
Extractability risk could be mitigated by restricting or limiting the API interface. For example: 
(1) \emph{API access control}, assuming queries can be attributed to individuals, limiting maximum $m$ serves as a natural defense, as indicated by \Cref{fig:extraction} (Left). It may also be possible to detect sequences of queries that appear to be probe attempts, and rate-limit or block such queries. 
(2) \emph{Decoding constraints}, suppressing the sampling probability~\cite{basu2020mirostat} or interrupting when overlap is detected~\cite{ippolito2022preventing} can directly reduce extractability. This can be combined with limits on the available decoding parameters that can be selected by users.

\shortsection{Evaluation and Discussion}
We evaluate both potential empirical mitigation strategies and differentially private training in \Cref{fig:defense}.
For empirical defense evaluation, we include instruction-tuning (with the Tulu-3~\cite{lambert2024tulu} model tuned on Llama-3.1-8B~\cite{grattafiori2024llama}) and a recently proposed post-training ParaPO~\cite{chen2025parapo} (with direct preference optimization tuned on the both Llama-3.1-8B and Tulu-3). We use their open-released model checkpoints~\cite{chen2025parapo} and use the BookSum~\cite{kryscinski2021booksum} as the protected target set.
For differentially-private training, we fine-tune GPT-2 on Enron~\cite{klimt2004enron} with DP noise using various privacy loss budgets.

Our observations are: 
(1) a combination of these empirical training-phase defenses can amplify the extractability mitigation. For example, compared to the original Llama model, Tulu w/ ParaPO reduces the extractable portion given $b=32$ from $20\%$ to $0.4\%$.
(2) combining API access control with both empirical and theoretical defenses further suppresses extractability with minimal effort. For example, when API limits are paired with Tulu tuned with ParaPO only $0.03\%$ of samples are found to be  extractable.
(3) DP training can reduce the extraction ratio. However, varying $\epsilon$ does not lead to reductions across several orders of magnitude, especially for shorter sequences such as $l=10$ (see \Cref{APP:fig:defense} in the Appendix).
This reflects our \Cref{theo:dpd_de} and highlights the need for future work to better understand connections between training noise and effective extraction defense.

\section{Conclusion}

We proved formally and demonstrated empirically that extraction risk for API-served large language models is a separate security objective that is not captured by  classical distinguishability games underlying DP-based defenses and general inference privacy attacks. 
The reducibility chain often implicitly assumed does not hold from interpreting extraction risk with standard privacy tools; consequently, low MIA advantage or a tight DP budget does not necessarily ensure low extraction risk. 
To address this gap, we formalized a worst-case extraction game and introduced $(l,b)$-\emph{inextractability} that upper bounds the expected number of trials required for verbatim suffix regeneration. 
Our estimator accounts for prefix selection, optimal decoding configuration, and multiple trials, yielding a conservative per-sample bound that can be related to the magnitude of realistic attack resources (queries, tokens, dollars, and time). 
The framework provides a tight extraction risk estimation for greedy generation, upper bounds for probabilistic generation, and can be extended to untargeted goal and approximate criterion, enabling unified auditing.
Empirical analyses across models reveal that vulnerable profiles are distinct for MIA and data extraction and provide a guideline for understanding the extraction cost.

\section*{Acknowledgments}
\noindent 
We would like to thank reviewers for their constructive comments and efforts in improving our paper. 
This work was supported in part by National Science Foundation grants (CNS-2437345, CNS-2125530, CNS-2124104, IIS-2302968), National Institutes of Health grants (R01LM013712, R01ES033241), and funds
provided by the National Science Foundation, Department of Homeland Security, and IBM through the ACTION AI Institute (Award \#2229876).
Any opinions, findings, and conclusions or recommendations expressed in this material are those of the
authors and do not necessarily reflect the views
of the National Science Foundation or its federal
agency and industry partners.

\bibliographystyle{IEEEtran}
\bibliography{ref,bib/llm}

\section*{Ethics considerations}
\noindent
This work proposes a metric to quantify extraction risk in LLMs. All experiments are conducted on publicly available models and datasets; no private data or models are used, and no human subjects are involved. The paper fully discloses the measurement methodology. Potential risks are discussed at a conceptual level only, without revealing actionable attacks on deployed systems.

\section*{LLM usage considerations}
\noindent
\textit{Originality}: LLMs were used for editorial purposes in this manuscript, and all outputs were inspected by the authors to ensure accuracy and originality.
\textit{Transparency}: All LLMs evaluated in this work are publicly available open-source models. No ideas or methodological components were introduced by LLMs; all conceptual contributions originate from the authors.
\textit{Responsibility}: All experiments rely exclusively on publicly available datasets and open-source model releases. No additional data was collected, and no proprietary resources were used. We selected models with at most 8B parameters to run on NVIDIA A100/H100 GPUs and to minimize environmental and computational overhead. 
We reused existing publicly released checkpoints for transparency and reproducibility.

\appendices

\section{Additional Results}\label{APP:sec:add_results}

\shortsection{Estimation of local log-Lipschitz regularity}
We provide the local log-lipschitz regularity estimation in \Cref{APP:fig:lipschitz} for approximate distance metrics including cosine similarity of embedding, n-gram based distance ($1-$ROUGE-L), and edit distance, which covers existing metrics used for approximate extraction measurement~\cite{ippolito2022preventing,hayes2024measuring,kassem2024alpaca}.

\begin{figure*}
\centering
\includegraphics[width=0.3\linewidth]{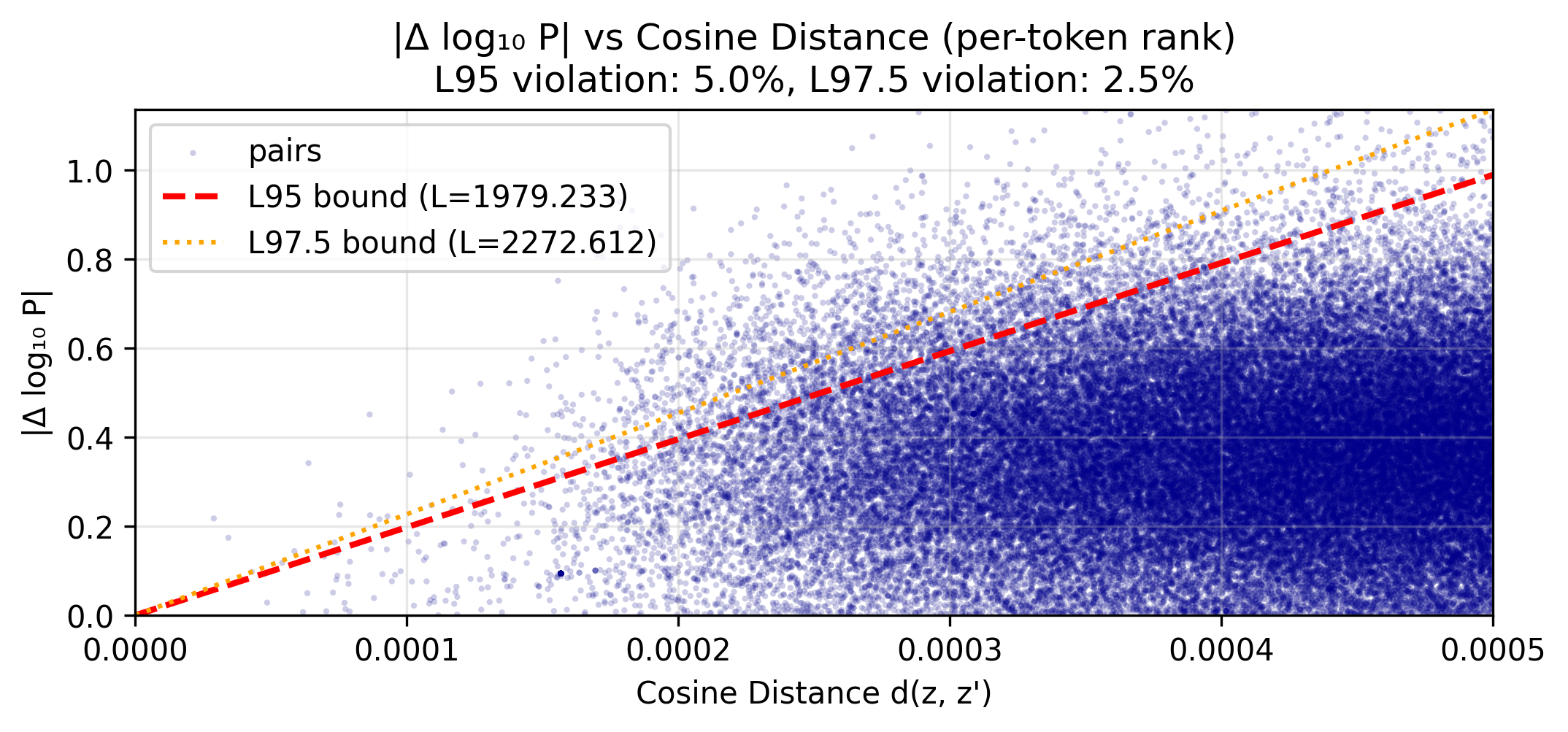}
\includegraphics[width=0.3\linewidth]{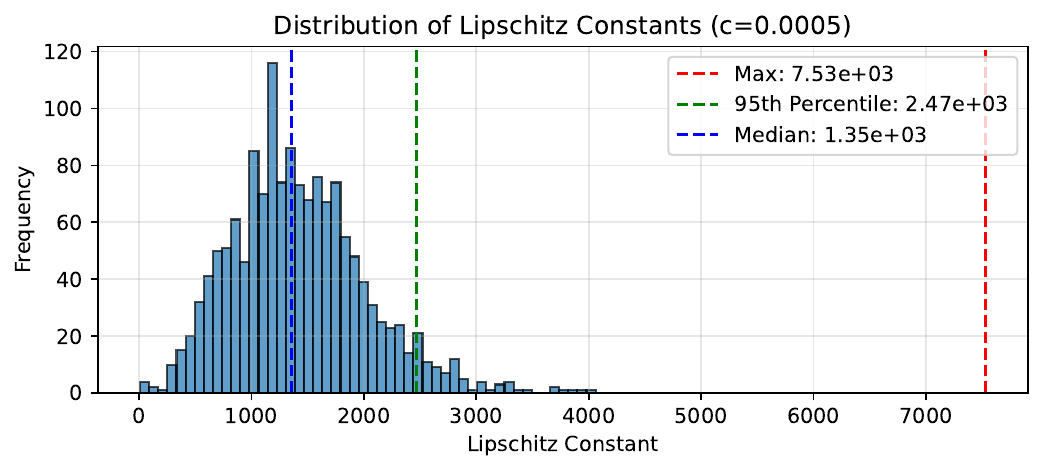}
\includegraphics[width=0.34\linewidth]{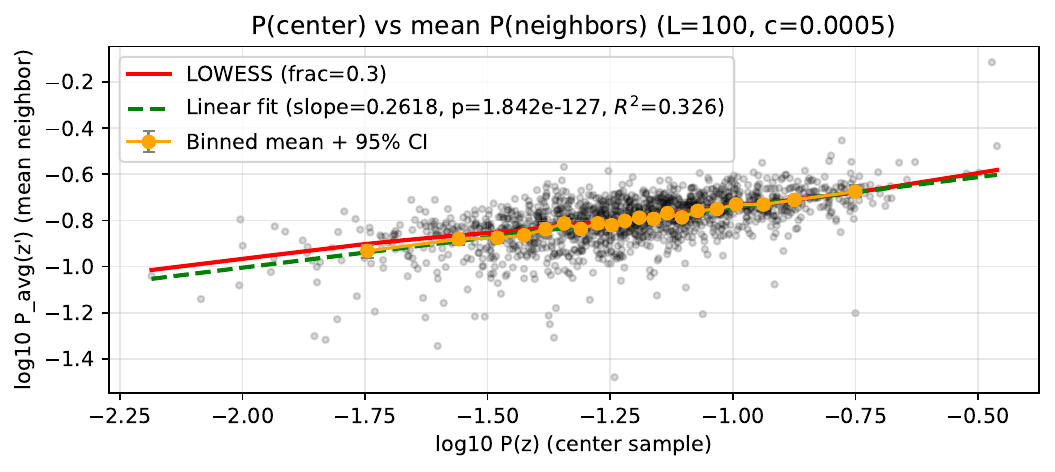}

\includegraphics[width=0.3\linewidth]{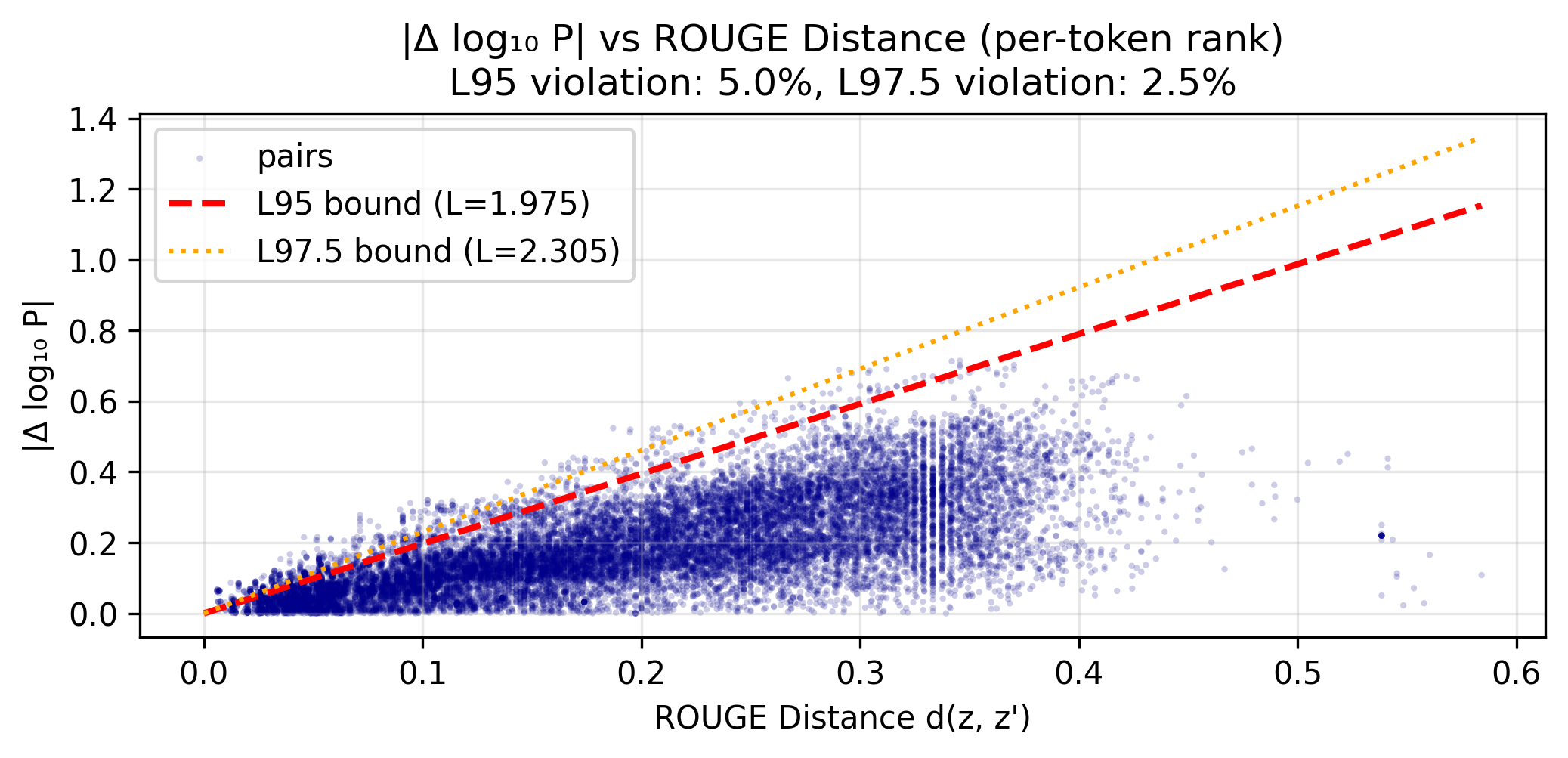}
\includegraphics[width=0.3\linewidth]{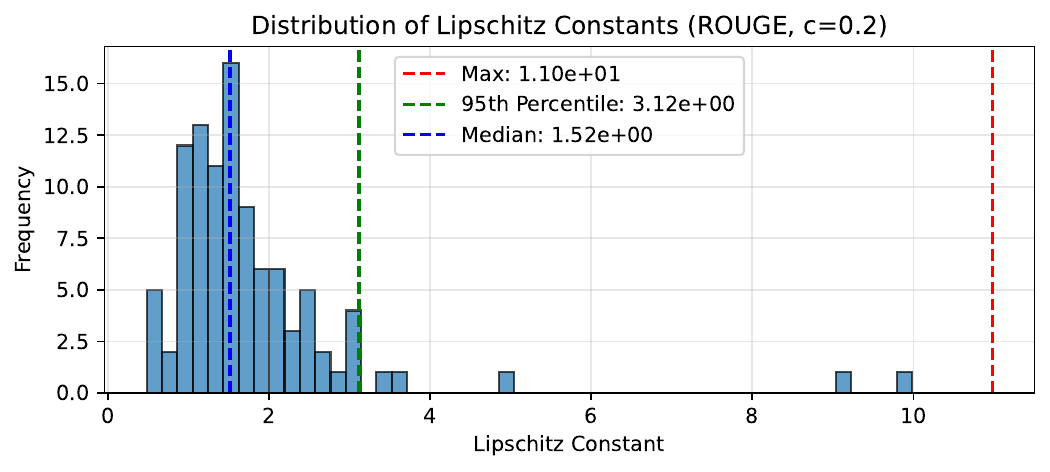}
\includegraphics[width=0.34\linewidth]{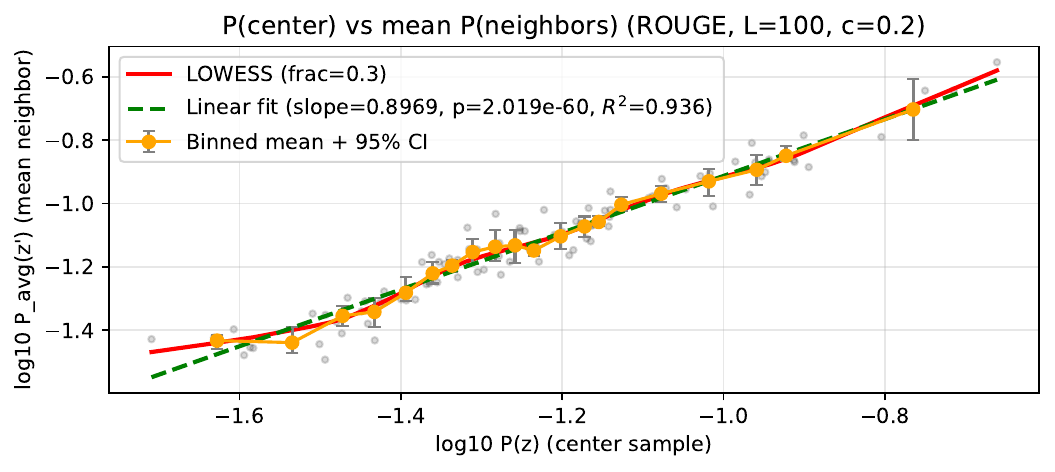}

\includegraphics[width=0.3\linewidth]{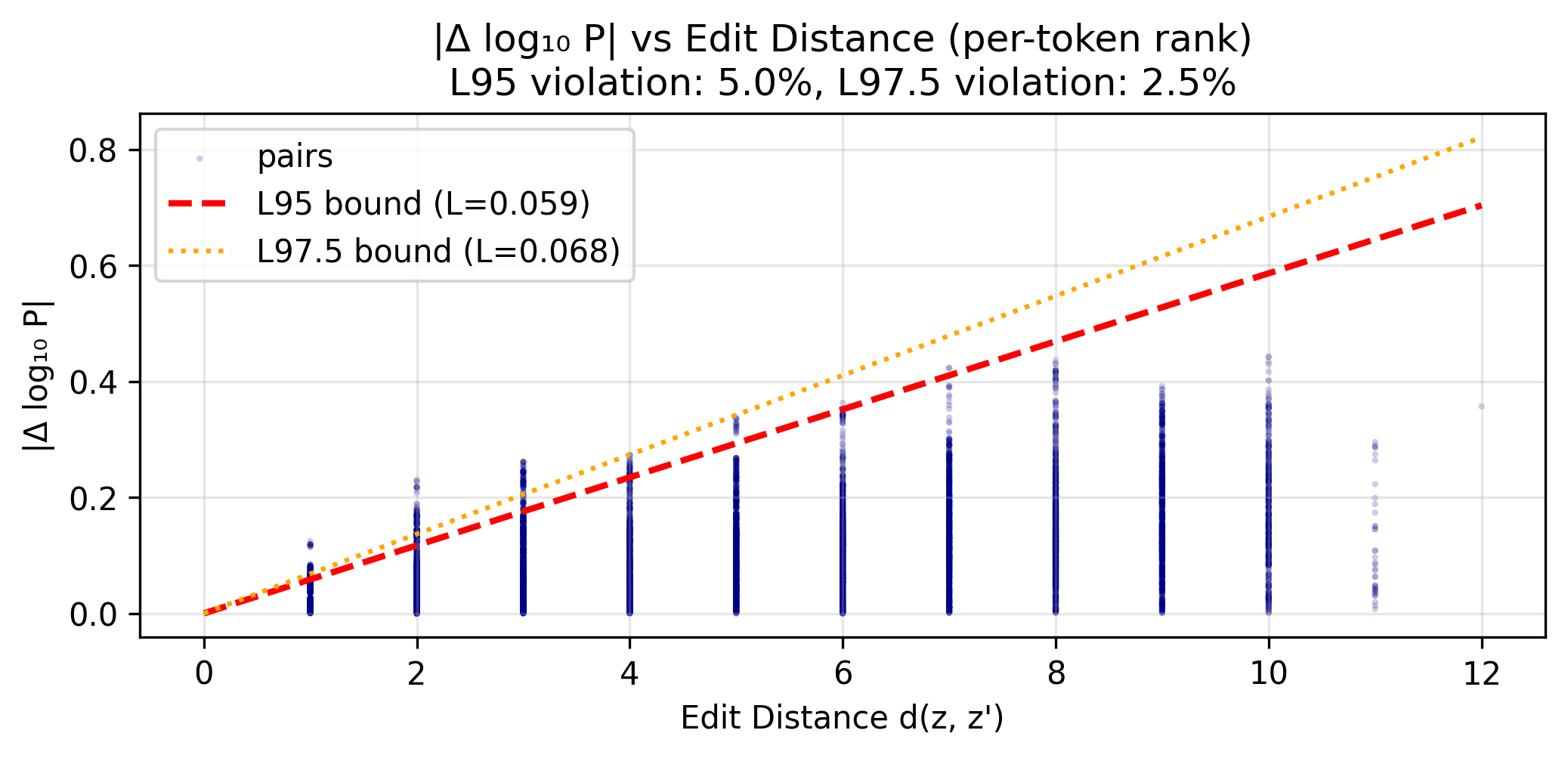}
\includegraphics[width=0.3\linewidth]{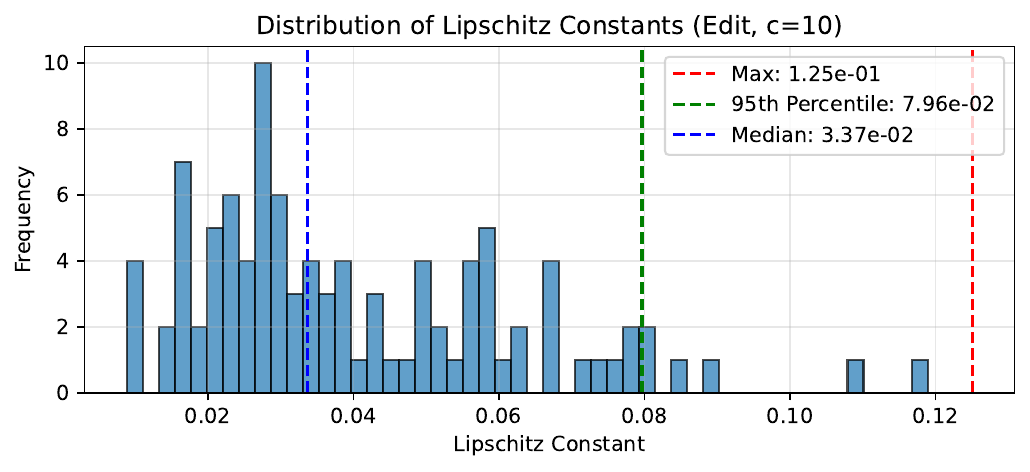}
\includegraphics[width=0.34\linewidth]{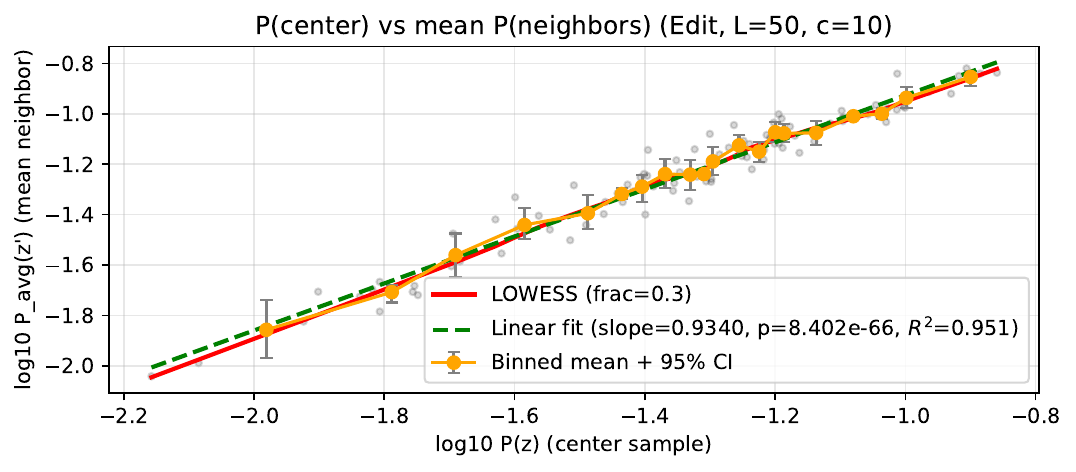}

\caption{Additional Empirical Evaluation of Local Log-Lipschitz Regularity on embedding cosine distance, 1-ROUGE-L, and edit distance.}
\label{APP:fig:lipschitz}
\end{figure*}

\shortsection{Correlation between text properties}
We analyze the spearman correlation between every two properties mentioned in \Cref{sec:mia_ext}, along with the MIA risk indicator of `Ref-Loss' and inextractability level $b$ in \Cref{APP:fig:heatmap_properties}.

\shortsection{Excluding MIA strength as a confounder}
\Cref{APP:fig:auc_correlation} supports \Cref{fig:mia_ext_corr} by showing that the correlation between each sample's indistinguishability and inextractability does not depend on the MIA strength for a specific MIA signal.

\begin{figure}
\centering
\includegraphics[width=0.8\linewidth]{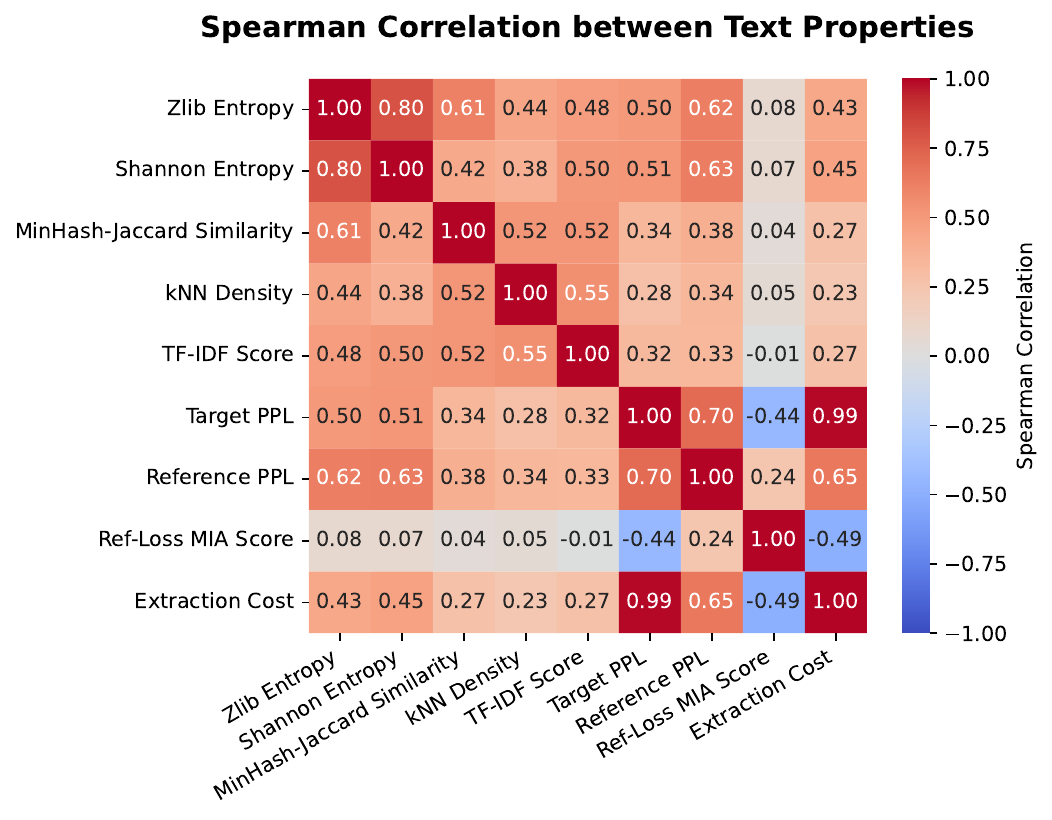}
\caption{Correlation between text property metrics.}
\label{APP:fig:heatmap_properties}
\end{figure}

\shortsection{Evaluation on different data domain}
\Cref{APP:fig:domain} evaluates the extraction cost for different domains in pre-training data with $l=50$.
We find that the code and book domain have the higher extractable rate than other domains, demonstrating the potential risks of copyright infringement.

\begin{figure}
\centering
\includegraphics[width=0.8\linewidth]{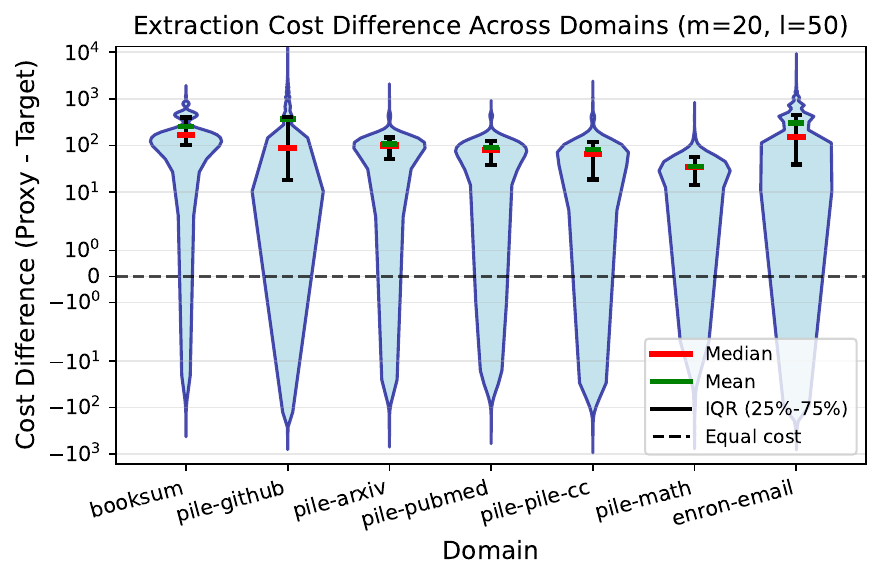}
\caption{Extraction cost across pre-training data domains, where the target model is Llama-3.1-8B and the proxy model is GPT-2.}
\label{APP:fig:domain}
\end{figure}

\begin{figure}
    \centering
    \includegraphics[width=0.8\linewidth]{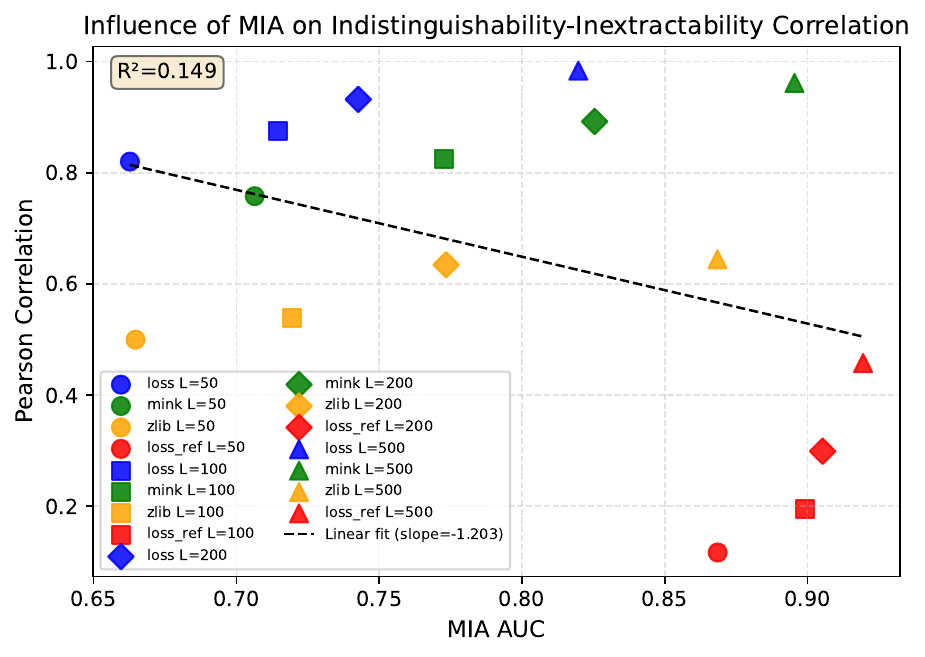}
    \caption{The correlation between indistinguishability and inextractability for each sample depends on whether the MIA signal is calibrated as shown in \Cref{fig:mia_ext_corr}, instead of the capability of the MIA signal (MIA AUC).
    For example, `Loss-Ref' has the highest MIA AUC and has the lowest correlation; `Zlib' has lower MIA AUC than `MinK' but also has lower correlation than `MinK'.
    }
    \label{APP:fig:auc_correlation}
\end{figure}

\begin{figure*}[htb]
 	\centering
 	\includegraphics[width=\linewidth]{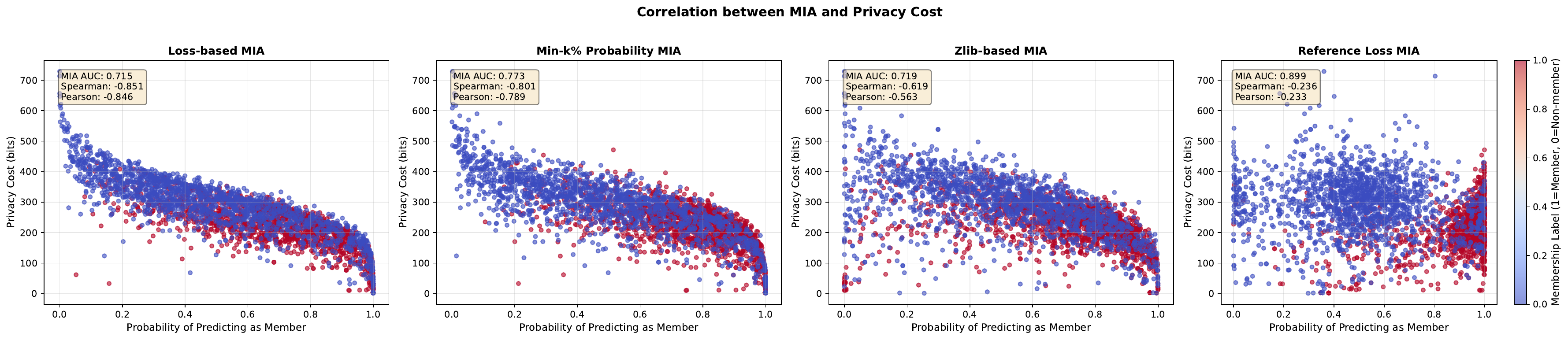}
 	\caption{Correlation of extraction cost and MIAs with various signals of GPT-2 fine-tuned on Enron dataset.}
 	\label{APP:fig:mia_ext}
\end{figure*}

\shortsection{Extraction evaluation with defenses}
We demonstrate the shorter sequence setting for DP defense evaluation in \Cref{APP:fig:defense}, which has significantly higher extraction rate than longer sequences in \Cref{fig:defense} under the same level of defense.
While this is not surprising given the probability production, the ineffectiveness of smaller $\epsilon$ is also more obvious, highlighting the incapability of distinguishability defense on extraction attack, especially for shorter text.

\begin{figure}
\centering
\includegraphics[width=\linewidth]{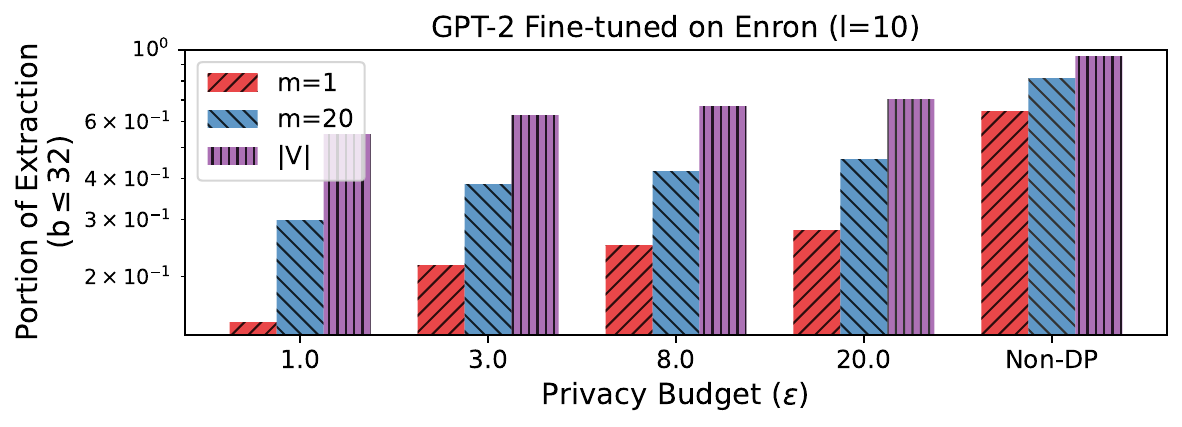}
\caption{Evaluation of DP training with $l=10$.}
\label{APP:fig:defense}
\end{figure}

\section{Proofs and Additional Analysis}\label{sec:appendix_proofs}
\subsection{Proof of Lemma~\ref{lemma:independent_optimal}}
\textit{Proof.}
By definition of $p_{\mathbf{z}}^*$, for every single-attempt strategy the success probability is bounded by $p_{\mathbf{z}}^*$. For any $n$-attempt strategy (possibly adaptive) let $q_j$ denote the conditional success probability of the $j$-th attempt given the history of previous attempts; then $q_j \le p_{\mathbf{z}}^*$ for all $j$. The overall success probability of the $n$-attempt strategy equals
$$
1-\prod_{j=1}^n (1-q_j) \le 1-\prod_{j=1}^n (1-p_{\mathbf{z}}^*) = 1-(1-p_{\mathbf{z}}^*)^n,
$$
which yields the claimed upper bound. The bound is tight, and can be achieved by repeating an optimal single-attempt strategy independently $n$ times. %

\subsection{Proof of Theorem~\ref{thm:targeted_to_untargeted}}
\textit{Proof.}
Let $\mathcal{S}=\{\mathbf{z}\in D_{\text{pro}}:|\mathbf{z}|=l\}$ with $|\mathcal{S}|=M$. By $(l,b)$-extraction privacy (Definition~\ref{def:ext_bit}), for every $\mathbf{z}\in\mathcal{S}$ the optimal single-attempt extraction probability satisfies $p_{\mathbf{z}}^*\le 2^{-b}$. The event of untargeted success is the union $\bigcup_{\mathbf{z}\in\mathcal{S}}\{\text{emit }\mathbf{z}\}$. Applying the union bound,
$$
\Pr\big[\exists \mathbf{z}\in\mathcal{S}:\text{ emit }\mathbf{z}\big] \le \sum_{\mathbf{z}\in\mathcal{S}} p_{\mathbf{z}}^* \le M 2^{-b}.
$$
Since probabilities cannot exceed $1$, we obtain $\Pr[\exists \mathbf{z}] \le \min\{1, M2^{-b}\}$. Defining $b_{\mathrm{un}} := -\log_2 \Pr[\exists \mathbf{z}\in\mathcal{S}:\text{ emit }\mathbf{z}]$ and using $-\log_2(\min\{1,x\}) = \max\{0,-\log_2 x\}$ for $x>0$ gives
$$
b_{\mathrm{un}} \ge \max\{0, b - \log_2 M\}.
$$
\noindent\textbf{Independence refinement.} If, in addition, the per-$\mathbf{z}$ emission events are (conditionally) independent, then
$$
\Pr[\exists \mathbf{z}\in\mathcal{S}] = 1-\prod_{\mathbf{z}\in\mathcal{S}} (1-p_{\mathbf{z}}^*) \le 1-(1-2^{-b})^{M},
$$
yielding the tighter bound $b_{\mathrm{un}} \ge -\log_2\big(1-(1-2^{-b})^{M}\big) \ge b-\log_2 M$ when $M2^{-b}\ll 1$. This completes the proof.

\subsection{Proof of Theorem~\ref{thm:worst_case_extraction_risk}}\label{app:worst_case_extraction_risk}
\textit{Proof.} Consider any position $t$ with ground-truth token $z_t$ of rank $r_t$ in the descending order of the original (temperature $T=1$) distribution $\mathbf{p}_t$. Under a sampling strategy with truncation size $k$ and temperature $T$, the probability of sampling $z_t$ is
$$
\Pr[\hat{z}_t \mid k,T] = \begin{cases} \dfrac{\exp(s_{z_t}/T)}{\sum_{v\in S_k} \exp(s_v/T)} & z_t\in S_k, \\ 0 & z_t\notin S_k, \end{cases}
$$
where $S_k$ is the set of top-$k$ tokens after scaling. For any fixed $k\ge r_t$, by letting $T\to 0$ the distribution mass concentrates on the top-$r_t$ logits, making the (approximately) uniform distribution over those $r_t$ tokens the extremal case. Thus the maximal achievable sampling probability for $z_t$ does not exceed $1/r_t$, and this bound is (asymptotically) attained by choosing $k=r_t$ and letting $T\to 0$. Hence
$$
\max_{k,T} \Pr[\hat{z}_t \mid k,T] = \frac{1}{r_t}.
$$
For a suffix $\vz_s$ the optimal single-attempt strategy (teacher-forced context, optimal truncation and temperature per position) yields an upper bound given by the product of per-token maxima due to the chain rule,
$$
\max_{k,T} \Pr\big[\mathcal{O}_{f_\theta \circ\phi_{k,T}}(\vz_p) = \vz_s\big] \le \prod_{t=1}^{l} \frac{1}{r_t}.
$$
Repeating the optimal single-attempt strategy independently $n$ times, the probability of at least one success is bounded by the standard complement form
$$
1-\Big(1-\prod_{t=1}^{l} \frac{1}{r_t}\Big)^n,
$$
which yields the stated worst-case extraction risk bound.

\subsection{Approximate Neighborhood Analysis}\label{app:approx_neighbor}
\begin{lemma}[Neighborhood Multiplicative Concentration]\label{lem:neigh_concentration}
Under Assumption~\ref{ass:lipschitz}, for every $z'\in\mathcal{N}(z)$ we have
$$
e^{-L_0 c} P_v(z) \le P_v(z') \le e^{L_0 c} P_v(z).
$$
Letting $\mu := |\mathcal{N}(z)|^{-1} \sum_{z'\in\mathcal{N}(z)} P_v(z')$ yields
$$
e^{-L_0 c} \le \frac{\mu}{P_v(z)} \le e^{L_0 c}, \\
\Big|\frac{\mu}{P_v(z)}-1\Big| \le e^{L_0 c}-1 \approx L_0 c.
$$
\end{lemma}
\textit{Proof of Lemma~\ref{lem:neigh_concentration}.} Assumption~\ref{ass:lipschitz} gives $|\log P_v(z')-\log P_v(z)|\le L_0 d(z,z')\le L_0 c$, hence $\log P_v(z')\in[\log P_v(z)-L_0 c,\log P_v(z)+L_0 c]$ and the stated bounds. Averaging preserves inequalities, giving the mean concentration.
The approximation at the end assumes a small $L_0 c$, which is practical given our observation in \Cref{fig:beyond_prob_change}.

\textit{Proof of Corollary~\ref{cor:ratio_suppression}.}
Let pre/post log-Lipschitz constants be $L_0$ and $L_0'$ over the same radius $c$. 
From Lemma~\ref{lem:neigh_concentration} (pre) and its post-defense analogue, the log-Lipschitz assumption gives
$$
\mu \ge e^{-L_0 c}P_v(z), \qquad \mu' \le e^{L_0' c}P'_v(z).
$$
The defense goal of suppression $\mu'<\mu$ requires $e^{L_0' c}P'_v(z) < e^{-L_0 c}P_v(z)$, i.e.\ $P_v(z)/P'_v(z) > e^{(L_0+L_0')c}$. Substituting $P'_v(z)=P_v(z)2^{-\Delta_b}$ yields $2^{\Delta_b}>e^{(L_0+L_0')c}$, hence $\Delta_b > (L_0+L_0')c/\ln 2$.

\shortsection{Estimating the log-Lipschitz constants $L_0,L_0'$}
Given a target text $z$, sample $K$ neighbors $\{z'_k\}_{k=1}^K$ satisfying $d(z,z'_k)\le c$ and compute per-sample ratios
$$
L_k \;=\; \frac{|\log P_v(z'_k) - \log P_v(z)|}{d(z,z'_k)}.
$$
We set $L_0$ to the 95th percentile of $\{L_k\}$, i.e.\ the smallest slope such that the line $y=L_0\,x$ through the origin covers 95\% of the scatter of $(d(z,z'_k),\;|\log P_v(z'_k)-\log P_v(z)|)$.
This avoids the conservatism of a global maximum while providing a high-confidence bound.
The post-defense constant $L_0'$ can be  obtained identically using $P'_v$.

\end{document}